\documentclass[final,3p,times,twocolumn]{elsarticle}
\usepackage[fleqn]{amsmath}
\usepackage{multirow,bigstrut}
\usepackage{amssymb,amsthm,bm}
\usepackage{geometry}
\usepackage{algorithm}
\usepackage{array}
\usepackage{algpseudocode}
\usepackage[retainorgcmds]{IEEEtrantools}

\usepackage{graphicx}
\graphicspath{{figures-pdf/}}
\usepackage[normalem]{ulem}
\usepackage{url}
\usepackage{lineno}

\usepackage{bm}
\usepackage{graphicx,color,overpic}
\usepackage{float}

\newlength\imgwidth
\newlength\figwidth
\newlength\sfigwidth
\newlength\vfigskip
\usepackage{color}
\definecolor{dgreen}{rgb}{0,.6,0}

\newcommand\mymatrix[1]{\bm{\mathrm{#1}}}

\newtheorem{theorem}{Theorem}
\newtheorem{proposition}{Proposition}

\newtheorem{corollary}{Corollary}

\newtheorem{fact}{Fact}

\usepackage[font=small,skip=0pt]{caption}
\begin{document}

\begin{frontmatter}


\title{On the security defects of an image encryption scheme}
\author[hkcityu,hkpolyu]{Chengqing Li\corref{corr}}
\ead{DrChengqingLi@gmail.com}
\author[germany]{Shujun Li\corref{corr}}
\ead[url]{www.hooklee.com}
\author[malaysia]{Muhammad Asim}
\author[spain]{Juana Nunez}
\author[spain]{Gonzalo Alvarez}
\author[hkcityu]{Guanrong Chen}

\address[hkcityu]{Department of Electronic Engineering, City University of Hong Kong,
83 Tat Chee Avenue, Kowloon Tong, Hong Kong SAR, China}

\address[hkpolyu]{Department of Electronic and Information Engineering,
Hong Kong Polytechnic University, Hung Hom, Kowloon, Hong Kong SAR,
China}

\address[germany]{Fachbereich Informatik und Informationswissenschaft, \\Universit\"at
Konstanz, Universit\"atsstra{\ss}e 10, 78457 Konstanz, Germany}

\address[malaysia]{Universiti Teknologi PETRONAS, 31750, Tronoh, Perak, Malaysia}

\address[spain]{Instituto de F\'{\i}sica Aplicada, Consejo Superior de
Investigaciones Cient\'{\i}ficas, Serrano 144, 28006 Madrid, Spain}

\begin{abstract}
This paper studies the security of a recently-proposed chaos-based
image encryption scheme, and points out the following problems: 1)
there exist a number of invalid keys and weak keys, and some keys
are partially equivalent for encryption/decryption; 2) given one
chosen plain-image, a subkey $K_{10}$ can be guessed with a smaller
computational complexity than that of the simple brute-force attack;
3) given at most 128 chosen plain-images, a chosen-plaintext attack
can possibly break the following part of the secret key: $\{K_i\bmod
128\}_{i=4}^{10}$, which works very well when $K_{10}$ is not too
large; 4) when $K_{10}$ is relatively small, a known-plaintext
attack can be carried out with only one known plain-image to recover
some visual information of any other plain-images encrypted by the
same key.
\end{abstract}


\end{frontmatter}

\section{Introduction}

Spurred by the rapid development of multimedia and network
technologies, multimedia data are being transmitted over networks
more and more frequently. As a result, content protection of
multimedia data is urgently needed in many applications, including
both public and private services such as military information
systems and multimedia messaging systems (MMS). Although any
traditional data ciphers (such as DES and AES) can be used to meet
this increasing demand of information security, they cannot provide
satisfactory solutions to some special properties and requirements
in many multimedia-related applications. For example, one
requirement is perceptual encryption \cite{Li:PVEA:IEEETCASVT2007},
meaning that the encrypted multimedia data can still be decoded by
any standard-compliant codec and displayed, with a relatively low
quality, which cannot be realized by simply employing a traditional
cipher. As a response to this concern, a large number of
specially-designed multimedia encryption schemes have been proposed
\cite{Greek:SCANImageEncryption:JEI95,
Chuang:BaseSwitchingImageEncryption:JEI97,
Guo-Yen-Pai:HDSP:IEEPVISP2002, Chen&Yen:RCES:JSA2003,
Chung&Chang:SCANImageEncryption:PRL98,
YaobinMao:3DBakerImageEncryption:IJBC2004,YaobinMao:CSF2004}.
Meanwhile, security analysis on the proposed schemes have also been
developed, and some of these schemes have been found insecure to a
certain extent \cite{Li:AttackingPOMC2004, Li:AttackingRCES2004,
Li:AttackingCNN2004, Li:BreakingBS:2006,
Li:AnalyzingHDSP:IEEPVISP2006,Li:AttackingISWBE2006,Li:AttackingBSSE2006}.
For more discussions about multimedia data encryption techniques,
readers are referred to some recent surveys
\cite{Li:ChaosImageVideoEncryption:Handbook2004,
Furht:ImageVideoEncryption:Handbook2004,
Uhl:ImageVideoEncryption:Book2005,
Furht:MultimediaSecurity:Book2005,
Zeng:MultimediaSecurity:Book2006}.

Since 2003, Pareek et al. \cite{Pareek:PLA2003, Pareek:CNSNS2005,
Pareek:ImageEncrypt:IVC2006} have proposed three different
encryption schemes based on one or more one-dimensional chaotic
maps, among which the one proposed in
\cite{Pareek:ImageEncrypt:IVC2006} was designed for image
encryption. Recent cryptanalysis results \cite{Alvarez:PLA2003,
Li:Pareek:CSF2008} have shown that the two schemes proposed in
\cite{Pareek:PLA2003, Pareek:CNSNS2005} are not secure. The present
paper focuses on the security analysis of the image encryption
scheme proposed in \cite{Pareek:ImageEncrypt:IVC2006}, and reports
the following findings:

\begin{enumerate}
\item
There are several types of security problems with the secret key,
and each subkey is involved in at least one problem.

\item
One subkey $K_{10}$ can be separately searched with a relatively
small computational complexity, even when only one chosen
plain-image is given.

\item
The scheme is insecure against chosen-plaintext attack in the sense
that using 128 chosen plain-images may be enough to break part of
the key. The attack is especially feasible when $K_{10}$ is not too
large.

\item
When $K_{10}$ is relatively small and one plain-image is known, a
known-plaintext attack can be used to reveal some visual information
of any other plain-images encrypted with the same secret key.
\end{enumerate}

The rest of the paper is organized as follows. The next section
gives a brief introduction to the image encryption scheme under
study. Section \ref{sec:Cryptanalysis} is the main body of the
paper, focusing on a comprehensive cryptanalysis, with both
theoretical and experimental results. In the last section, some
concluding remarks and conclusions are given.

\section{The image encryption scheme under study}
\label{sec:scheme}

In this scheme, the plaintext is a color image with separate RGB
channels. The plain-image is scanned in the raster order, and then
divided into 16-pixel blocks. The encryption and decryption
procedures are performed blockwise on the plain-image. Without loss
of generality, assume that the size of the plain-image is $M\times
N$, and that $MN$ can be exactly divided by 16. Then, the
plain-image $\bm{I}$ can be represented as a 1-D signal
$\{I(i)\}_{i=0}^{MN-1}$ with $N_b=MN/16$ blocks, namely,
$\bm{I}=\{I^{(16)}(k)\}_{k=0}^{N_b-1}$, where
$I^{(16)}(k)=\{I(16k+i)\}_{i=0}^{15}$. Similarly, the cipher-image
is denoted by $\bm{I}^*=\{I^{*(16)}(k)\}_{k=0}^{N_b-1}$, where
$I^{*(16)}(k)=\{I^*(16k+i)\}_{i=0}^{15}$.

The secret key of the encryption scheme under study is an 80-bit
integer and can be represented as $K=K_1\cdots K_{10}$, where each
subkey $K_i\in\{0,\ldots,255\}$. Two chaotic systems are involved in
the encryption scheme, and both are realized by iterating the
Logistic map
\begin{equation}
f(x)=\mu x(1-x),
\label{equation:logistic}
\end{equation}
where $\mu$ is the control parameter and fixed to be 3.9999. One
chaotic map runs globally throughout the whole encryption process,
while another one runs locally for the encryption of each 16-pixel
block. The initial condition of the global chaotic map is determined
by the six subkeys $K_4\sim K_9$ as follows:
\begin{equation}
\label{equation:X0}
X_{0}=\left(\frac{\sum_{i=4}^{6}K_i\cdot
2^{8(i-4)}}{2^{24}}+\frac{\sum_{j=7}^9((K_j\bmod 16)+\lfloor
K_j/16\rfloor)}{96}\right)\bmod 1,
\end{equation}
and the local chaotic map corresponding to each block is initialized
according to selected chaotic states of the global map. For the
$k$-th block $I^{(16)}(k)$, the encryption process can be described
by the following steps.

\begin{itemize}
\item
\textit{Step 1: Determining the initial condition of the local
chaotic map.} Iterate the global chaotic map until 24 chaotic states
within the interval $[0.1,0.9)$ are obtained. Denoting these chaotic
states by $\{\hat{X}_j\}_{j=1}^{24}$, generate 24 integers
$\{P_j\}_{j=1}^{24}$, where $P_j=\lfloor
24(\hat{X}_j-0.1)/0.8\rfloor+1$.\footnote{In Sec.~2 of
\cite{Pareek:ImageEncrypt:IVC2006}, the interval is $[0.1,0.9]$ and
$P_j=\lfloor 23(\hat{X}_j-0.1)/0.8\rfloor+1$. However, following
this process, $P_j=24$ when and only when $\hat{X}_j=0.9$, which
becomes a rare event and conflicts with the requirement that $P_j$
has a roughly uniform distribution over $\{1,\ldots,24\}$.
Therefore, in this paper we changed the original process in
\cite{Pareek:ImageEncrypt:IVC2006} to a more reasonable one. Note
that such a change does not affect the performance of the encryption
scheme.} Then, calculate $B_2=\sum_{i=1}^{3}K_i\cdot 2^{8(i-1)}$ and
set the initial condition of the local chaotic map as
\begin{equation}\label{eq:generateY0}
Y_0=\left(\frac{B_2+\sum_{j=1}^{24}B_2[P_j]\cdot
2^{j-1}}{2^{24}}\right)\bmod 1,
\end{equation}
where $B_2[P_j]$ denotes the $P_j$-th bit of $B_2$.

\item
\textit{Step 2: Encrypting the $k$-th block $I^{(16)}(k)$.} For each
pixel in the block, iterate the local chaotic map to obtain $K_{10}$
consecutive chaotic states $\{\hat{Y}_j\}_{j=1}^{K_{10}}$ which fall
into the interval [0.1,0.9), and then encrypt the RGB values of the
current pixel according to the following formulas:
\begin{eqnarray}
R^* =
E_1(R)=g_{K_4,K_5,K_7,K_8,\hat{Y}_{K_{10}}}\circ\cdots\circ
g_{K_4,K_5,K_7,K_8,\hat{Y}_1}(R),\label{equation:encryptR}\\
G^* =
E_2(G)=g_{K_5,K_6,K_8,K_9,\hat{Y}_{K_{10}}}\circ\cdots\circ
g_{K_5,K_6,K_8,K_9,\hat{Y}_1}(G),\label{equation:encryptG}\\
B^* =
E_3(B)=g_{K_6,K_4,K_9,K_7,\hat{Y}_{K_{10}}}\circ\cdots\circ
g_{K_6,K_4,K_9,K_7,\hat{Y}_1}(B),\label{equation:encryptB}
\end{eqnarray}
\normalsize where $\circ$ denotes the composition of two functions
and $g_{a_0,b_0,a_1,b_1,Y}(x)$ is a function under the control of
$Y$ as shown in Table~\ref{table:encryption}.

\item
\textit{Step 3: Updating subkeys $K_1,\ldots,K_9$.} Perform the
following updating operation for $i=1\sim 9$:
\begin{equation}\label{eq:update}
K_i=\left(K_i+K_{10}\right)\bmod 256.
\end{equation}
\end{itemize}

\begin{table*}[htbp]
\center \caption{The definition of $g_{a_0,b_0,a_1,b_1,Y}(x)$, where
$\overline{x}$ denotes the bitwise complement of $x$, and $\oplus$
denotes the bitwise XOR operation.}\label{table:encryption}
\scriptsize
\begin{tabular}{c|c|c}
\hline $Y\in$ & $g_{a_0,b_0,a_1,b_1,Y}(x)$= & $g_{a_0,b_0,a_1,b_1,Y}^{-1}(x)$=\\
\hline $[0.10,0.13)\cup[0.34,0.37)\cup[0.58,0.62)$ &
\multicolumn{2}{c}{$\overline{x}=x\oplus 255$}\\
\hline $[0.13,0.16)\cup [0.37, 0.40) \cup [0.62, 0.66)$ &
\multicolumn{2}{c}{$x\oplus a_0$}\\
\hline $[0.16,0.19)\cup[0.40, 0.43) \cup [0.66, 0.70)$ &
$(x+a_0+b_0)\bmod 256$ & $(x-a_0-b_0)\bmod 256$\\
\hline $[0.19,0.22)\cup [0.43, 0.46) \cup [0.70, 0.74)$ &
\multicolumn{2}{c}{$\overline{x\oplus a_0}=x\oplus(a_0\oplus 255)=x\oplus\overline{a_0}$}\\
\hline $[0.22,0.25)\cup [0.46, 0.49) \cup [0.74, 0.78)$ &
\multicolumn{2}{c}{$x\oplus a_1$}\\
\hline $[0.25,0.28)\cup [0.49, 0.52) \cup [0.78, 0.82)$ &
$(x+a_1+b_1)\bmod 256$ & $(x-a_1-b_1)\bmod 256$\\
\hline $[0.28,0.31)\cup [0.52, 0.55) \cup [0.82, 0.86)$ &
\multicolumn{2}{c}{$\overline{x\oplus a_1}=x\oplus(a_1\oplus 255)=x\oplus\overline{a_1}$}\\
\hline $[0.31,0.34)\cup[0.55,0.58)\cup[0.86,0.90]$ &
\multicolumn{2}{c}{$x=x\oplus 0$}\\\hline
\end{tabular}
\end{table*}

The decryption procedure is similar to the above encryption
procedure, except that Eqs.
(\ref{equation:encryptR})$\sim$(\ref{equation:encryptB}) in Step 2
are replaced by the following ones:
\begin{eqnarray}
R=
E_1^{-1}(R^*)=g_{K_4,K_5,K_7,K_8,\hat{Y}_1}^{-1}\circ\cdots\circ
g_{K_4,K_5,K_7,K_8,\hat{Y}_{K_{10}}}^{-1}(R^*),\\
G=
E_2^{-1}(G^*)=g_{K_5,K_6,K_8,K_9,\hat{Y}_1}^{-1}\circ\cdots\circ
g_{K_5,K_6,K_8,K_9,\hat{Y}_{K_{10}}}^{-1}(G^*),\\
B=
E_3^{-1}(B^*)=g_{K_6,K_4,K_9,K_7,\hat{Y}_1}^{-1}\circ\cdots\circ
g_{K_6,K_4,K_9,K_7,\hat{Y}_{K_{10}}}^{-1}(B^*),
\end{eqnarray}
where $g_{a_0,b_0,a_1,b_1,Y}^{-1}(x)$ is the inverse function of
$g_{a_0,b_0,a_1,b_1,Y}(x)$ with respect to $x$ as shown in
Table~\ref{table:encryption}.

\section{Cryptanalysis}
\label{sec:Cryptanalysis}

In this section, we report our cryptanalysis results about the image
encryption scheme under study. These include a comprehensive
analysis on invalid keys, weak keys and partially equivalent keys, a
chosen-plaintext attack to break $K_{10}$, a chosen-plaintext attack
to break $\{K_i\bmod 128\}_{i=4}^{10}$, a known-plaintext attack,
and some other minor security problems.

\subsection{Two properties of the scheme}
\label{ssec:basicproperties}

To facilitate the description of the discussion below, we first
point out two properties of the scheme under study in this
subsection. One is about the subkey updating mechanism, and the
other is about the essential equivalent presentation form of the
encryption function.

To improve the security of the scheme, an updating mechanism is
introduced for subkeys in Eq.~(\ref{eq:update}) of
\cite{Pareek:ImageEncrypt:IVC2006}. Because the updating process is
performed in a finite-state field, the sequence of each updated
subkey produced by such a mechanism is always periodic (see
Fact~\ref{fact:T_K10} below). As a result, the sequence of the
dynamic keys is also periodic. Assuming that the period is $T$, the
$N_b$ plain pixel-blocks $\{I^{(16)}(k)\}_{k=0}^{N_b-1}$ can be
divided into $T$ separate sets according to the values of these
dynamically updated subkeys:
$\left\{\mathbb{I}_j=\bigcup\limits_{k=0}^{N_T-1} I^{(16)}(T\cdot
k+j)\right\}_{j=0}^{T-1}$, where $N_T=\lceil N_b/T\rceil$. For
blocks in the same set $\mathbb{I}_j$, all the updated subkeys are
identical. In other words, for each set $\mathbb{I}_j$ ($1/T$ of the
whole plain-image) one can consider that the secret key is fixed.
Since $1/T$ of a plain-image may be enough to reveal essential
visual information, one can turn to break any set $\mathbb{I}_j$
without considering the updating mechanism.

\begin{fact}\label{fact:T_K10}
For $x,a\in\{0,\ldots,255\}$, the integer sequence
$\{y(i)=(x+ai)\bmod256\}_{i=0}^{\infty}$, has period
$T=256/\gcd(a,256)$.
\end{fact}

With respect to the encryption function, one can see from
Table~\ref{table:encryption} that each encryption subfunction is
represented in one of the following two formats:
\begin{enumerate}
\item
$g_{a_0,b_0,a_1,b_1,Y}(x)=x\oplus\alpha$, where
$\alpha\in\{0,255,a_0,a_1,\overline{a_0},\overline{a_1}\}$;

\item
$g_{a_0,b_0,a_1,b_1,Y}(x)=x\dotplus\beta$, where $x\dotplus \gamma$
denotes $(x+\gamma)\bmod 256$ (the same hereinafter), and
$\beta\in\{a_0\dotplus b_0,a_1\dotplus
b_1\}\subset\{0,\cdots,255\}$.
\end{enumerate}

Because
$(x\oplus\alpha_1)\oplus\alpha_2=x\oplus(\alpha_1\oplus\alpha_2)$
and
$(x\dotplus\beta_1)\dotplus\beta_2=x\dotplus(\beta_1\dotplus\beta_2)$,
consecutive encryption subfunctions of the same kind can be combined
together, and those with $\alpha=0$ or $\beta=0$ can be simply
ignored. As a result, each encryption function $E_i(x)$ is a
composition of $len\leq K_{10}$ subfunctions:
$\{G_j(x)\}_{j=1}^{len}$, where $G_j(x)=x\oplus\alpha_{\lceil
j/2\rceil}$ or $x\dotplus\beta_{\lceil j/2\rceil}$, and $G_j(x)$,
$G_{j+1}(x)$ are encryption subfunctions of different kinds.
According to the types of $G_1(x)$ and $G_{len}(x)$, $E_i(x)$ has
four different formats:

\begin{enumerate}
\item
$E_i(x)=((\cdots((x\dotplus\beta_1)\oplus\alpha_1)\cdots)\oplus\alpha_{\lceil
(len-1)/2\rceil})\dotplus\beta_{\lceil len/2\rceil}$;

\item
$E_i(x)=((\cdots((x\dotplus\beta_1)\oplus\alpha_1)\cdots)\dotplus\beta_{\lceil
(len-1)/2\rceil})\oplus\alpha_{\lceil len/2\rceil}$;

\item
$E_i(x)=((\cdots((x\oplus\alpha_1)\dotplus\beta_1)\cdots)\oplus\alpha_{\lceil
(len-1)/2\rceil})\dotplus\beta_{\lceil len/2\rceil}$;

\item
$E_i(x)=((\cdots((x\oplus\alpha_1)\dotplus\beta_1)\cdots)\dotplus\beta_{\lceil
(len-1)/2\rceil})\oplus\alpha_{\lceil len/2\rceil}$.
\end{enumerate}

Note that $len$ is generally less than $K_{10}$. Assuming that
$\{Y_i\}$ distributes uniformly over the interval [0.1,0.9], we can
get the following inequality:
\begin{equation}
Prob[len=K_{10}]\leq
\begin{cases}
2\cdot (\frac{5}{8}\cdot\frac{1}{4})^{\frac{K_{10}}{2}}, &
\mbox{when $K_{10}$ is even},\\
(\frac{5}{8}\cdot\frac{1}{4})^{\left\lfloor\frac{K_{10}}{2}\right\rfloor}(\frac{5}{8}+\frac{1}{4}),
& \mbox{when  $K_{10}$ is odd.}
\end{cases}
\end{equation}
From the above equation, we can see that the probability decreases
exponentially as $K_{10}$ increases. Because it is difficult to
exactly estimate the probability that $len$ is equal to a given
value less than $K_{10}$, we performed a number of random
experiments for a $512\times 512$ plain-image to investigate the
possibilities. Figure \ref{figure:distributelengthfunction} shows a
result of 100 random keys when $K_{10}=66$.

\begin{figure}[htbp]
\centering
\begin{minipage}[t]{1.5\figwidth}
\centering
\includegraphics[width=1.5\figwidth]{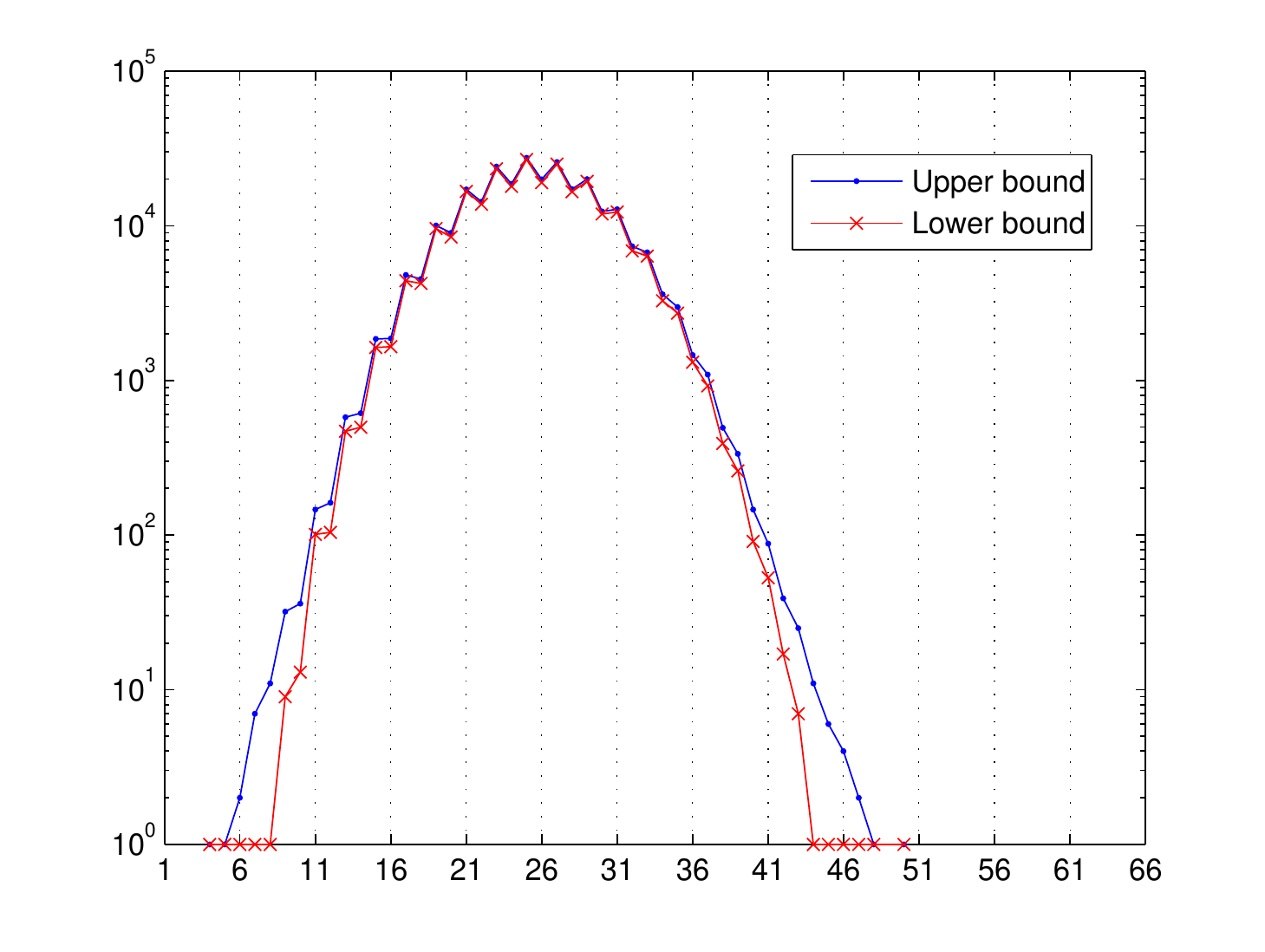}
\end{minipage}
\caption{The number of subfunctions composed of $len$
subfunctions, when $K_{10}=66$ and other subkeys were generated
randomly for 100 times.} \label{figure:distributelengthfunction}
\end{figure}

Since $G_j(x)$ is a composition of multiple functions
$g_{a_0,b_0,a_1,b_1,Y}(x)$ of the same kind, and since that
$\overline{a_0}\oplus a_1=a_0\oplus\overline{a_1}=a_0\oplus
a_1\oplus 255$ and $\overline{a_0}\oplus\overline{a_1}=a_0\oplus
a_1$, one can easily deduce that
\begin{multline}\label{eq:alpha_i_set}
\alpha_i\in\mathbb{A}=\{255, a_0, a_1, a_0\oplus 255, a_1\oplus
255, a_0\oplus a_1, \\a_0\oplus a_1\oplus 255\}
\end{multline}
and
\begin{align*}
\beta_i\in\mathbb{B} & =\{z_1(a_0\dotplus b_0)\dotplus
z_2(a_1\dotplus b_1)\;|\; z_1,z_2\in\{0,\cdots,K_{10}\} \\
 {} &        \mbox{ and } z_1+z_2\leq K_{10}\}.
\end{align*}

Note that $\mathbb{A}$ has an interesting property: $\forall
x_1,x_2\in\mathbb{A}\cup\{0\}$, $x_1\oplus
x_2\in\mathbb{A}\cup\{0\}$. This property concludes that
$\bigoplus_i\alpha_i\in\mathbb{A}\cup\{0\}$, which will be used
later in Sec.~\ref{ssec:CPA} for chosen-plaintext attack.

\subsection{Analysis of the key space}

In this subsection, we report some \textit{invalid keys},
\textit{weak keys} and \textit{partially equivalent keys} existing
in the encryption scheme under study. Here, an \textit{invalid key}
means a key that cannot ensure the successful working of the
encryption scheme, a \textit{weak key} is a key that corresponds to
one or more security defects, and \textit{partially equivalent keys}
generate the same encryption result for a certain part of the
plain-image. When estimating the key space, invalid keys and weak
keys should be excluded, and all keys that are partially equivalent
to each other should be counted as one single key \cite[Sec.
3.2]{AlvarezLi:Rules:IJBC2006}.

\subsubsection{Invalid keys with respect to $K_4\sim K_{9}$}

When $X_0=0$, the global chaotic map will fall into the fixed point
0, which disables the encryption process due to the lack of chaotic
states lying in $[0.1,0.9]$. Now, let us see when $X_0=0$ can
happen.

Observing Eq.~(\ref{equation:X0}), one can see that $X_0=0$ is
equivalent to
\begin{multline}
\frac{\sum_{i=4}^{6}K_i\cdot
2^{8(i-4)}}{2^{24}}\equiv\\
-\mathrm{FP}\left(\frac{\sum_{j=7}^9((K_j\bmod
16)+\lfloor K_j/16\rfloor)}{96}\right)\pmod 1
\end{multline}
where $\mathrm{FP}(x)$ denotes the floating-point value of $x$.
Because $0\leq\sum_{i=4}^{6}K_i\cdot 2^{8(i-4)}<2^{24}$ and
$0\leq\sum_{j=7}^9((K_j\bmod 16)+\lfloor K_j/16\rfloor)\leq 15\cdot
6=90<96$, one can further simplify the above equation as follows:
\begin{equation}\label{eq:key-X0equal0}
\frac{\sum_{i=4}^{6}K_i\cdot
2^{8(i-4)}}{2^{24}}=1-\frac{\mathrm{FP}\left(\sum_{j=7}^9((K_j\bmod
16)+\lfloor K_j/16\rfloor)\right)}{96}.
\end{equation}
By the fact that $\frac{\sum_{i=4}^{6}K_i\cdot
2^{8(i-4)}}{2^{24}}\bmod 2^{-24}=0$, the following equality also
holds:
\[
\frac{\mathrm{FP}\left(\sum_{j=7}^9((K_j\bmod 16)+\lfloor
K_j/16\rfloor)\right)}{96}\bmod 2^{-24}=0.
\]
By checking all the 91 possible values of $\sum_{j=7}^9((K_j\bmod
16)+\lfloor K_j/16\rfloor)$, one can easily get the following
result:
\begin{equation}\label{eq:key-X0equal0-P1}
\sum_{j=7}^9((K_j\bmod 16)+\lfloor K_j/16\rfloor)=3C,
\end{equation}
where $C\in[0, 30]$. In this case,
\[
1-\mathrm{FP}\left(\frac{\sum_{j=7}^9((K_j\bmod 16)+\lfloor
K_j/16\rfloor)}{96}\right)=1-\frac{C}{32}.
\]
Substituting the above equation into Eq.~(\ref{eq:key-X0equal0}),
one has
\begin{equation}\label{eq:key-X0equal0-P2}
\sum_{i=4}^{6}K_i\cdot 2^{8(i-4)}=2^{19}(32-C).
\end{equation}
As a result, any key that satisfies Eqs.~(\ref{eq:key-X0equal0-P1})
and (\ref{eq:key-X0equal0-P2}) simultaneously can lead to $X_0=0$.
The number of such invalid subkeys $(K_4,\cdots,K_9)$ can be
calculated to be $5592406=2^{22.415}$, where $5592406=\lceil
16^6/3\rceil$ is the number of distinct values of $(K_7,K_8,K_9)$
satisfying Eq.~(\ref{eq:key-X0equal0-P1}), calculated according to
the following Proposition~\ref{proposition:sum_mod3}.

\begin{proposition}\label{proposition:sum_mod3}
Given an $n$-dimensional vector
$\mathbf{A}=(a_1,\cdots,a_n)\in\{0,\cdots,15\}^n$, the number of
distinct values of $\mathbf{A}$ that satisfy $(a_1+\cdots+a_n)\bmod
3=0$, $1$ and $2$ are $\lceil 16^n/3\rceil$, $\lfloor 16^n/3\rfloor$
and $\lfloor 16^n/3\rfloor$, respectively.
\end{proposition}
\begin{proof}
This proposition can be proved by mathematical induction.

When $n=1$, one can easily verify that the number of distinct values
of $\mathbf{A}$ that satisfy $a_1\bmod 3=0,1,2$, are $6, 5, 5$,
respectively. Since $6=\lceil 16/3\rceil$ and $5=\lfloor
16/3\rfloor$, the proposition is true.

Assuming that the position is true for $1\leq n\leq k$, we prove the
case for $n=k+1$. First, rewrite $a_1+\cdots+a_{k+1}$ as
$A_k+a_{k+1}$, where $A_k=a_1+\cdots+a_k$. Then, observe that
$(A_k+a_{k+1})\bmod 3=0$ is equivalent to $A_k\equiv -a_{k+1}\pmod
3$. Thus, the number of distinct values of $\mathbf{A}$ that
satisfying $A_k+a_{k+1}\bmod 3=0$ is the following sum:
\begin{IEEEeqnarray*}{rCl}
\IEEEeqnarraymulticol{3}{l}{N[(A_k+a_{k+1})\bmod 3=0]}\\
& = & \lceil 16^k/3\rceil\cdot\lceil
16/3\rceil+2\lfloor 16^k/3\rfloor\cdot\lfloor 16/3\rfloor\\
& = & (\lfloor 16^k/3\rfloor+1)\cdot\lceil 16/3\rceil+2\lfloor
16^k/3\rfloor\cdot\lfloor 16/3\rfloor\\
& = & 16\cdot\lfloor 16^k/3\rfloor+6.
\end{IEEEeqnarray*}

Assume $16^k=(15+1)^k=3C+1$. Then, $16^{k+1}=48C+16$ and $\lceil
16^{k+1}/3\rceil=16C+\lceil 16/3\rceil=16C+6$. Then $16\cdot\lfloor
16^k/3\rfloor+6=16C+6=\lceil 16^{k+1}/3\rceil$. Going through a
similar process, one can easily get $N[(A_k+a_{k+1})\bmod
3=1]=N[(A_k+a_{k+1})\bmod 3=2]=\lfloor 16^{k+1}/3\rfloor$. This
completes the mathematical induction, hence finishes the proof of
the proposition.
\end{proof}

\subsubsection{Invalid keys with respect to $K_1\sim K_3$}

For a given block $I^{(16)}(k)$, if $Y_0=0$, the local chaotic map
will fall into the fixed point 0, which will also disable the
encryption process of the corresponding block. According to
Eq.~(\ref{eq:generateY0}), $Y_0=0$ when the following equality
holds:
\[
\left(B_2+\sum_{j=1}^{24}B_2[P_j]\cdot 2^{j-1}\right)\bmod 2^{24}=0,
\]
Since $0\leq B_2=\sum_{i=1}^3K_i\cdot 2^{8(i-1)}<2^{24}$ and
$0\leq\sum_{j=1}^{24}B_2[P_j]\cdot 2^{j-1}<2^{24}$, the above
equality can be simplified as follows:
\begin{equation}\label{equation:Y0equal0}
\sum_{j=1}^{24}B_2[P_j]\cdot 2^{j-1}=2^{24}-B_2.
\end{equation}
Assuming that $P_j$ distributes uniformly in $\{1,\cdots,24\}$,
$B_2$ and $(2^{24}-B_2)$ have $m$ and $n$ 0-bits, respectively, the
probability for Eq.~(\ref{equation:Y0equal0}) to hold is
\[
p_s=\left(\frac{m}{24}\right)^n\cdot\left(\frac{24-m}{24}\right)^{24-n}=\frac{m^n(24-m)^{24-n}}{24^{24}}.
\]
The relationship between the values of $p_s$ and $(25m+n)$ is shown
in Fig.~\ref{figure:probabilitysameY0}, from which one can see that
the probability is not negligible for some values of $(m,n)$. In
fact, because $p_s>0$ holds for any value of $(m,n)$, we can say
that any key is invalid from the strictest point of view. To resolve
this problem, the original encryption scheme must be amended. One
simple way to do so is setting $Y_0$ to be a pre-defined value once
$Y_0=0$ occurs. In the following discussions of this paper and all
experiments involved, we set $Y_0=1/2^{24}$ when such an event
occurs.

\begin{figure}[htbp]
\centering
\begin{minipage}[t]{1.5\figwidth}
\centering
\includegraphics[width=1.5\figwidth]{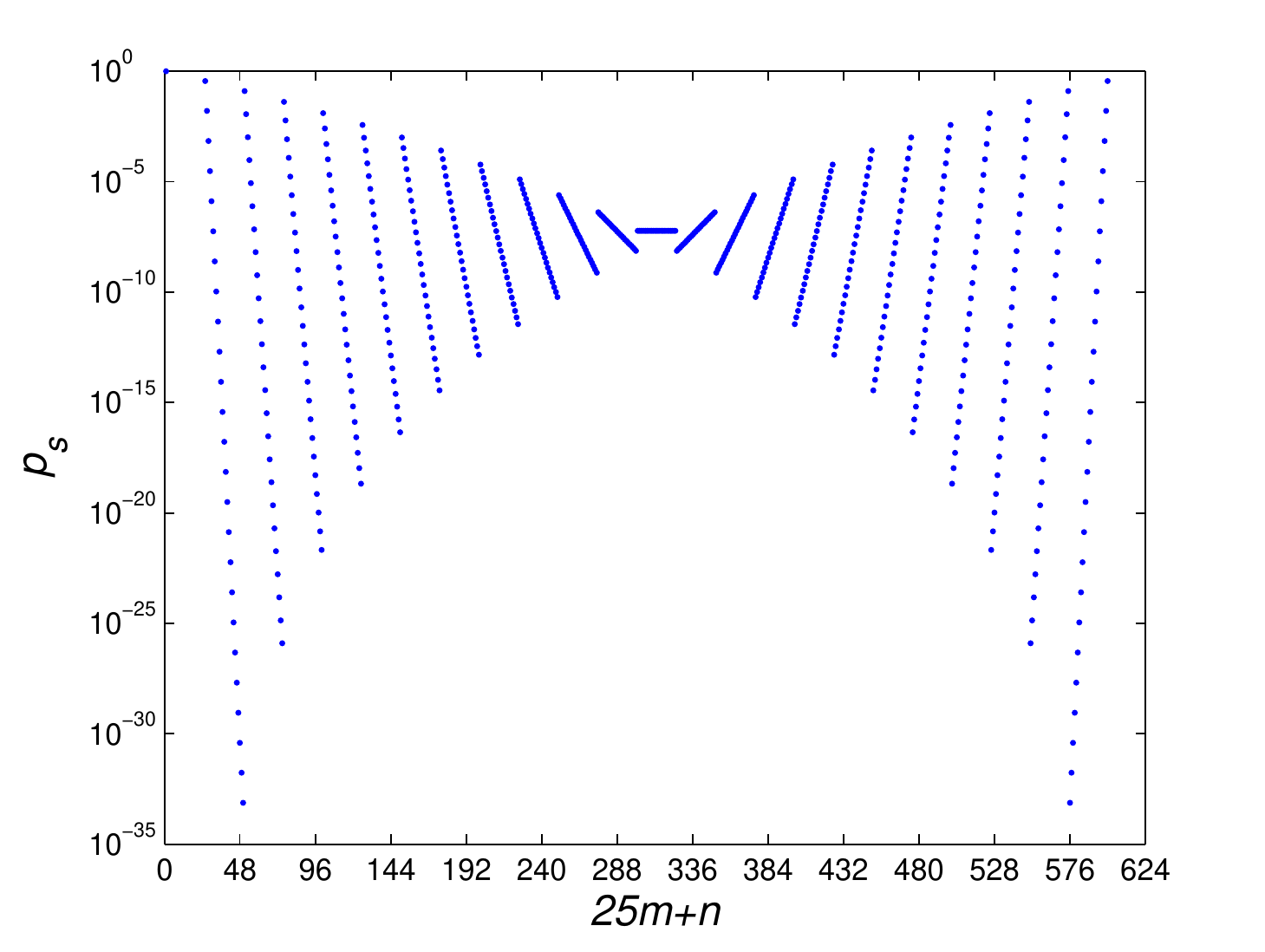}
\end{minipage}
\caption{The value of $p_s$ with respect to the value of $(25m+n)$,
where $m,n\in\{0,\cdots,24\}$.}\label{figure:probabilitysameY0}
\end{figure}

\subsubsection{Weak keys with respect to $K_{10}$}

In the encryption scheme under study, the update process of subkeys
$K_1\sim K_9$ and the number of subfunctions
$g_{a_0,b_0,a_1,b_1,Y}(x)$ in each encryption function are both
controlled by the subkey $K_{10}$. In the following, we discuss two
weak-key problems with respect to $K_{10}$, which correspond to the
above two processes controlled by $K_{10}$, respectively.

From Fact~\ref{fact:T_K10}, one can see that the update of subkeys
$K_1\sim K_9$ has an inherent weakness, i.e., the possible values
for the period of the sequence of the updated subkeys is $2^i$,
$i=1\sim 8$. For some values of $K_{10}$, this period can be very
small, which weakens the updating mechanism considerably. The worst
situation occurs when $K_{10}=128$, which corresponds to period two.
From the most conservative point of view, $T$ should take the
maximal value 256, which means that $K_{10}$ should be an odd
number.

The other problem deals with the number of subfunctions
$g_{a_0,b_0,a_1,b_1,Y}(x)$ in each encryption function. When
$K_{10}=1$, the probability for a pixel to remain unchanged is 1/8
(under the assumption that $Y_i$ distributes uniformly in the
chaotic interval). Though the probability seems quite large, our
experiments have shown that very little visual information leaks in
the cipher-image. When $K_{10}\geq 2$, experiments have shown that
it is almost impossible to distinguish any visual pattern from the
cipher-image. As a result, in this case there exists only one major
weak key: $K_{10}=1$. To avoid other potential security defects,
$K_{10}\geq 8$ is suggested.

\subsubsection{Weak keys with respect to $K_4\sim K_9$}
\label{ssection:weakkeysK4toK9}

Observing Table~\ref{table:encryption}, one can see that the
encryption subfunction $g_{a_0,a_1,b_0,b_1,y}$ $(x)=x$ or $\bar{x}$
when the following requirements are satisfied:
\begin{equation}\label{equation:weakkey_terms}
a_0,a_1\in\{0,255\}\mbox{ and }a_0+b_0\equiv a_1+b_1\equiv
0\pmod{256}.
\end{equation}
For the sub-image $\mathbb{I}_j$, if the subkeys corresponding to
one encryption function $E_i(x)$ satisfy the above requirements,
$E_i(x)$ will also be $x$ or $\bar{x}$. Assuming that the chaotic
trajectory of the local chaotic map has a uniform distribution in
the interval $[0.1, 0.9]$, the probability of
$g_{a_0,a_1,b_0,b_1,y}(x)=\bar{x}$ is $p=3/8$. Then, according to
Proposition~\ref{proposition:OddBinom} given below (note that
$\bar{x}=x\oplus 255$), $\forall i=1\sim 3$, the probabilities of
$E_i(x)=\bar{x}$ and $E_i(x)=x$ are $(1-(1/4)^{K_{10}})/2$ and
$(1+(1/4)^{K_{10}})/2$, respectively. This means that about half of
all plain-pixels in $\mathbb{I}_j$ are not encrypted at all, which
may reveal some visual information about the plain-image. As an
example, when $K=``3C1DE8FF0151FF012840"$ (which corresponds to
$T=4$), one of our experiments showed that 49.9\% of all the pixels
in $\mathbb{I}_0$ were not encrypted (see
Fig.~\ref{figure:weakkeyK4toK9} for the encryption result).

\begin{figure}[htbp]
\centering
\begin{minipage}[t]{\imgwidth}
\centering
\includegraphics[width=\textwidth]{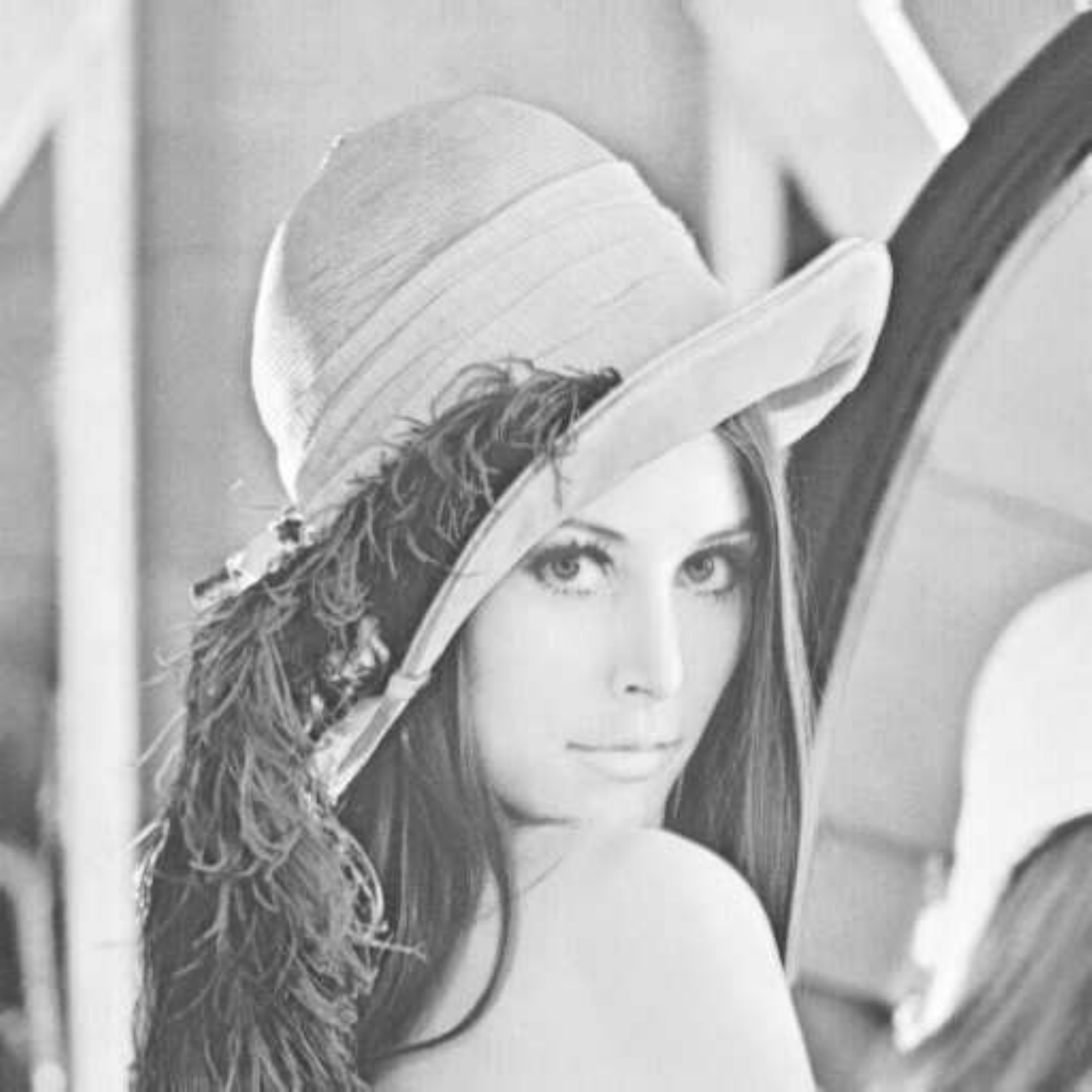}
a)
\end{minipage}
\begin{minipage}[t]{\imgwidth}
\centering
\includegraphics[width=\textwidth]{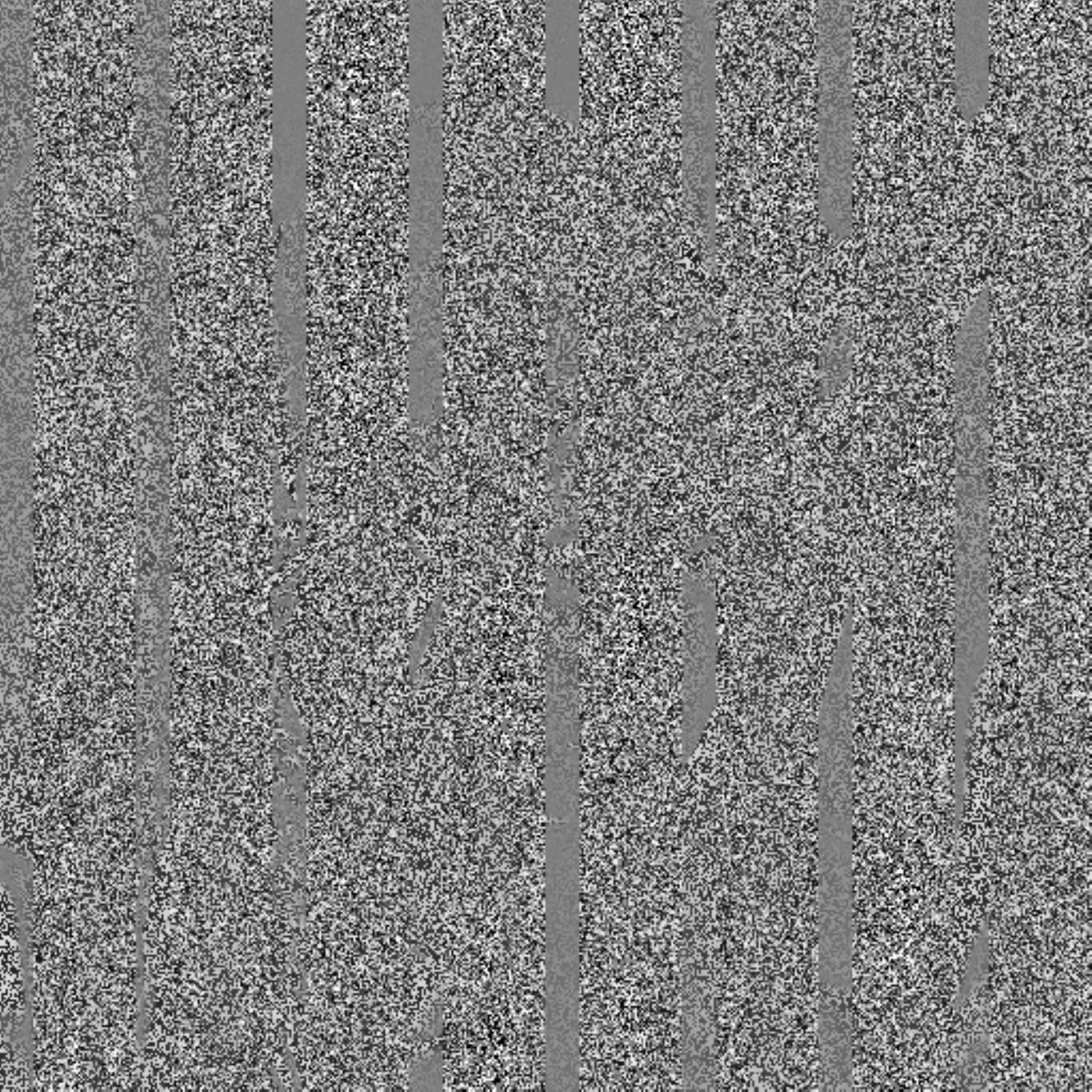}
b)
\end{minipage}
\caption{The encryption result when $K=``3C1DE8FF0151FF012840"$
(represented in hexadecimal format, the same hereinafter): a) the
red channel of the plain-image ``Lenna"; b) the red channel of the
cipher-image. For the other two color channels we have obtained
similar results.}\label{figure:weakkeyK4toK9}
\end{figure}

\begin{proposition}\label{proposition:OddBinom}
Given $n>1$ functions, $f_1(x),\ldots,f_n(x)$, assume that each
function is $x\oplus a$ with probability $p$ and is $x$ with
probability $1-p$, where $a\in\mathbb{Z}$. Then, the probability of
the composition function $F(x)=f_1\circ\cdots\circ f_n(x)=x\oplus a$
is $P=(1-(1-2p)^n)/2$.
\end{proposition}
\begin{proof}
Assume that $k=\lceil n/2\rceil$. Then, $n=2k$ if it is an even
integer and $n=2k-1$ when it is odd. To ensure
$F(x)=f_1\circ\cdots\circ f_n(x)=x\oplus a$, the number of
subfunctions that are equal to $x\oplus a$ should be an odd integer.
So,
\begin{eqnarray*}
P & = & \sum_{i=1}^k\binom{n}{2i-1}p^{2i-1}(1-p)^{n-(2i-1)}\\
& = &
(1-p)^n\cdot\sum_{i=1}^k\binom{n}{2i-1}(p/(1-p))^{2i-1}\\
& = &
(1-p)^n\cdot\frac{\left(1+p/(1-p)\right)^n-\left(1-p/(1-p)\right)^n}{2}\\
& = & (1-(1-2p)^n)/2.
\end{eqnarray*}
This completes the proof of the proposition.
\end{proof}

By letting Eq.~(\ref{equation:weakkey_terms}) hold for the three
encryption functions $E_1(x)$, $E_2(x)$ and $E_3(x)$, we found a
list of weak keys of this kind, as shown in
Table~\ref{table:WeakKeyLeaking}.

\begin{table*}[htbp]
\centering
\caption{Some weak keys that cause leaking of visual
information.}\label{table:WeakKeyLeaking}
\begin{tabular}{c|c}
\hline Weak keys & Visual information leaked from\\
\hline $(K_4,K_5),(K_7,K_8)\in\{(0,0),(255,1)\}$ & Channel R\\
\hline $(K_5,K_6),(K_8,K_9)\in\{(0,0),(255,1)\}$ & Channel G\\
\hline $(K_6,K_4),(K_9,K_7)\in\{(0,0),(255,1)\}$ & Channel B\\
\hline $(K_4,K_5,K_6,K_7,K_8,K_9)=(0,0,0,0,0,0)$ & the whole
plain-image\\\hline
\end{tabular}
\end{table*}

\subsubsection{Partially equivalent keys with respect to $K_7\sim K_9$: Class 1}

Observing Eq.~(\ref{equation:X0}), one can see that the value of
$X_0$ remains unchanged if the following segments of $K_7,K_8,K_9$
exchange their values: $K_7\bmod 16$, $\lfloor K_7/16\rfloor$,
$K_8\bmod 16$, $\lfloor K_8/16\rfloor$, $K_9\bmod 16$, $\lfloor
K_9/16\rfloor$. Now let us find out what will happen if we exchange
$K_9\bmod 16$ and $\lfloor K_9/16\rfloor$, i.e., exchange the upper
half and the lower half of $K_9$. In this case, since the encryption
of the red value of each pixel is independent of $K_9$, the red
channel of the cipher-image will remain unchanged. Similar results
also exist for $K_7$ and $K_8$, which correspond to unchanged blue
and green channels of the plain-image, respectively. This problem
reduces the subkey-space of $(K_7,K_8,K_9)$ from $256^3$ to
$(16+(256-16)/2)^3=136^3$.

\subsubsection{Partially equivalent keys with respect to $K_7\sim K_9$: Class 2}
\label{sssec:PartiallyEquivalentKey2}

As remarked in Sec.~\ref{ssec:basicproperties}, each encryption
subfunction $g_{a_0,a_1,b_0,b_1,Y}(x)$ can be represented in one of
the following two formats: $x\oplus \alpha$ and $x\dotplus \beta$.
The following two facts about $\oplus$ and $\dotplus$ will lead us
to construct another class of partially equivalent keys.

\begin{fact}
$\forall\ a\in\{0,\ldots,255\}$, $a\oplus 128=a\dotplus 128$.
\label{lemma:xor128}
\end{fact}

\begin{fact}\label{corollary:xor128}
$\forall\ a,b\in \mathbb{Z}$, $(a\oplus 128)\dotplus b=(a\dotplus
b)\oplus 128$.
\end{fact}

Fact~\ref{corollary:xor128} means that a change in the MSB (most
significant bit) of $x$, $a_0$, $a_1$, $b_0$, $b_1$ of any
encryption subfunction $g_{a_0,a_1,b_0,b_1,Y}(x)$ is equivalent to
XORing 128 on the output of the composition function $E_i(x)$.

Next, Fact~\ref{corollary:xor128} is used to figure out the second
class of partially equivalent keys about $K_7\sim K_9$. First,
choose any two subkeys from $K_7\sim K_9$. Without loss of
generality, let us take $K_7$ and $K_8$. Then, given a secret key
$K$ that satisfies $K_7<128$ and $K_8\geq 128$ (or, $K_7\geq 128$
and $K_8<128$), let us change it into another key $\widetilde{K}$ by
setting $\widetilde{K}_7=K_7\oplus 128$ and
$\widetilde{K}_8=K_8\oplus 128$. From Eq.~(\ref{equation:X0}), it is
easy to see that $X_0$ remains the same for the two keys. This means
that both the global and the local chaotic maps have the same
dynamics throughout the encryption procedure for the two keys, and
that the difference on ciphertexts is determined only by the
MSB-changes of $K_7$ and $K_8$. In the following, to analyze the
influence of the MSB-changes on the ciphertexts, we consider the
three color channels separately.

First, consider the encryption process of the green channel of the
plain-image, in which $K_7$ is not involved at all. Assuming that
the chaotic trajectory $\{Y_i\}$ distributes uniformly within the
interval $[0.1,0.9]$, the probability that $K_8$ has an effect on
each encryption subfunction is $p=3/8$. If $K_8$ appears for an even
number of times in the total $K_{10}$ encryption subfunctions, then
the value of $E_2(G)$ will remain the same for the two keys $K$ and
$\widetilde{K}$; otherwise, $E_2(G)$ changes its MSB. Thus, using
the same deduction as given in the proof of
Proposition~\ref{proposition:OddBinom}, the probability that
$E_2(G)$ remains unchanged can be calculated to be
$P_2=(1+(1-2p)^{K_{10}})/2=(1+4^{-K_{10}})/2$. This means that more
than half of all green pixel values in the ciphertexts are identical
in probability for the two keys $K$ and $\widetilde{K}$.

For the blue channel, $K_8$ is not involved in the encryption
process. So, following a similar deduction, the probability that
$E_3(B)$ remains unchanged is $P_3=(1+4^{-K_{10}})/2=P_2$.

For the red channel, both $K_7$ and $K_8$ are involved, but their
differences are neutralized for the encryption subfunction
$x\dotplus (K_7+K_8)$. So, the probability that the differences in
$K_7$ and $K_8$ have an effect on the ciphertext is reduced to be
$p=2/8=1/4$. Thus, the probability that $E_1(R)$ remains unchanged
becomes $P_1=(1+2^{-K_{10}})/2>P_2=P_3$.

Combining all the above analyses together, it is expected that more
than half of all pixel values in the cipher-images will be identical
for the two keys $K$ and $\widetilde{K}$. In addition, for other
different pixel values, the XOR difference is always equal to 128.
By enumerating all possibilities about this security problem, one
can conclude that the subkey-space of $(K_7,K_8,K_9)$ is reduced
from $256^3$ to $4\cdot 128^3=256^3/2$.

To verify the above theoretical results, we have carried out some
experiments for a plain-image of size $512\times 512$. One result is
shown in Fig.~\ref{figure:EquivalentKey2}, in which the number of
identical pixel values in red, green and blue channels are 131241
(50.06\%), 130864 (49.92\%) and 131383 (50.12\%), respectively.

\begin{figure*}[htbp]
\centering
\begin{minipage}[t]{\imgwidth}
\centering
\includegraphics[width=\textwidth]{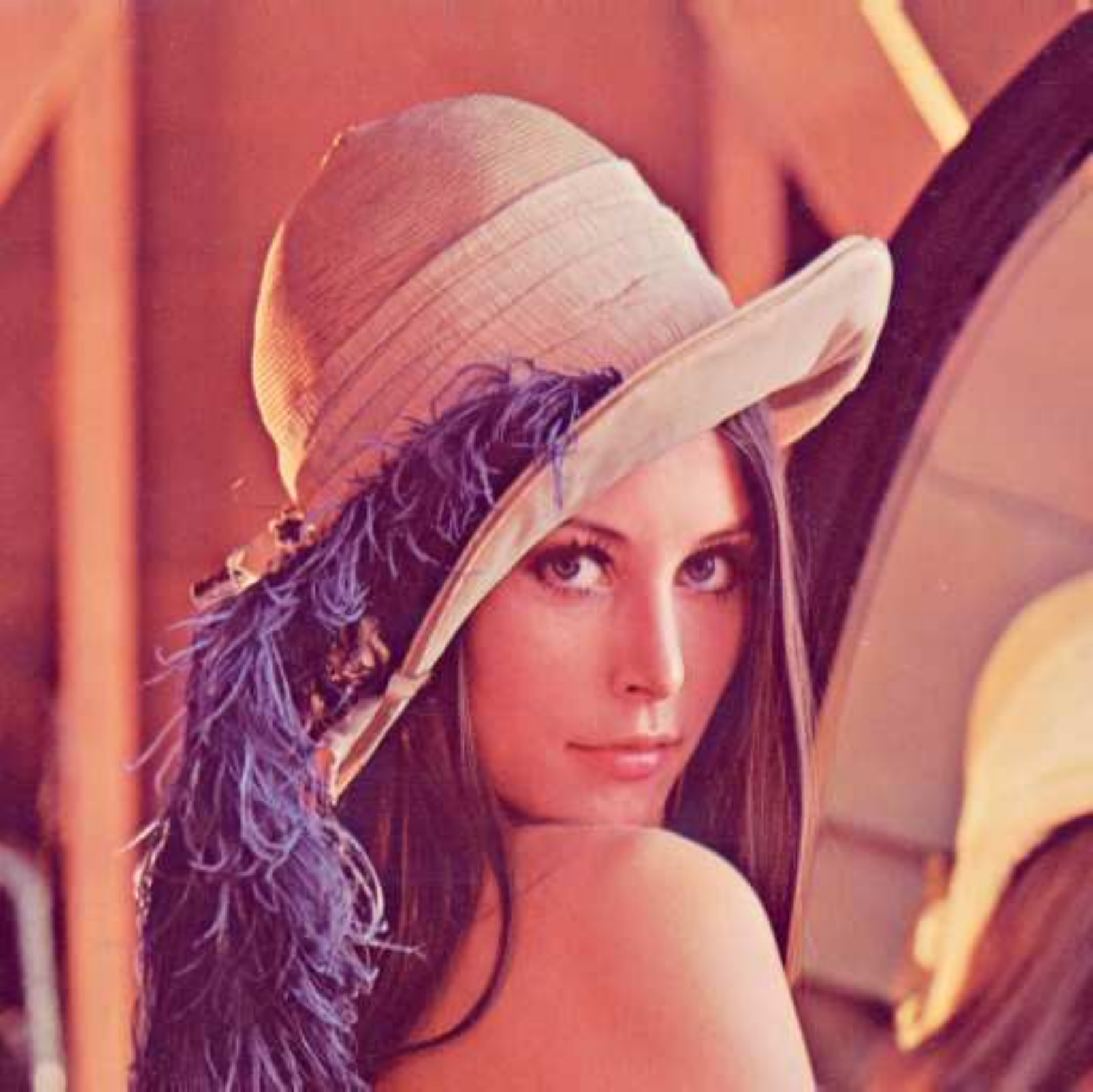}
a)
\end{minipage}
\begin{minipage}[t]{\imgwidth}
\centering
\includegraphics[width=\textwidth]{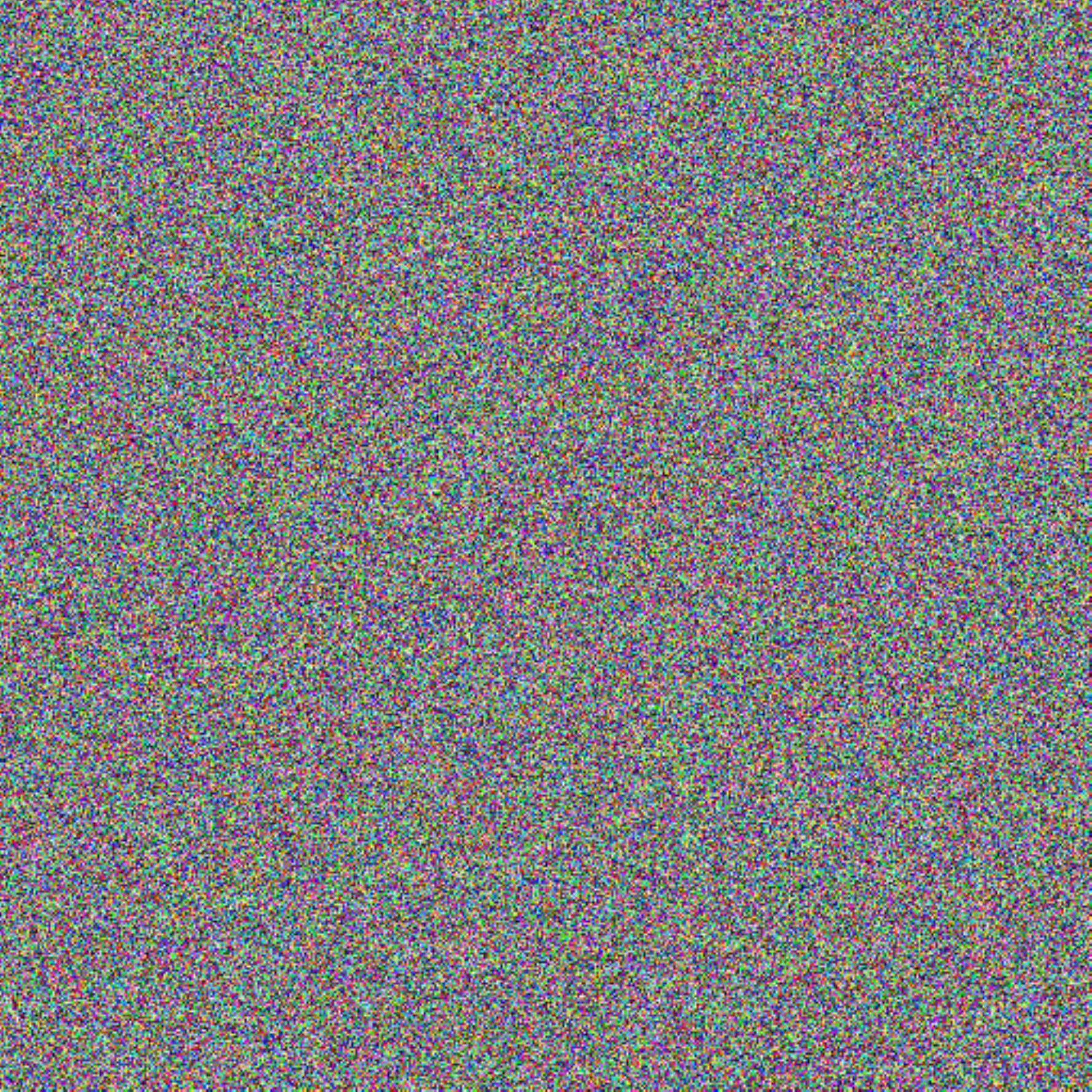}
b)
\end{minipage}
\begin{minipage}[t]{\imgwidth}
\centering
\includegraphics[width=\textwidth]{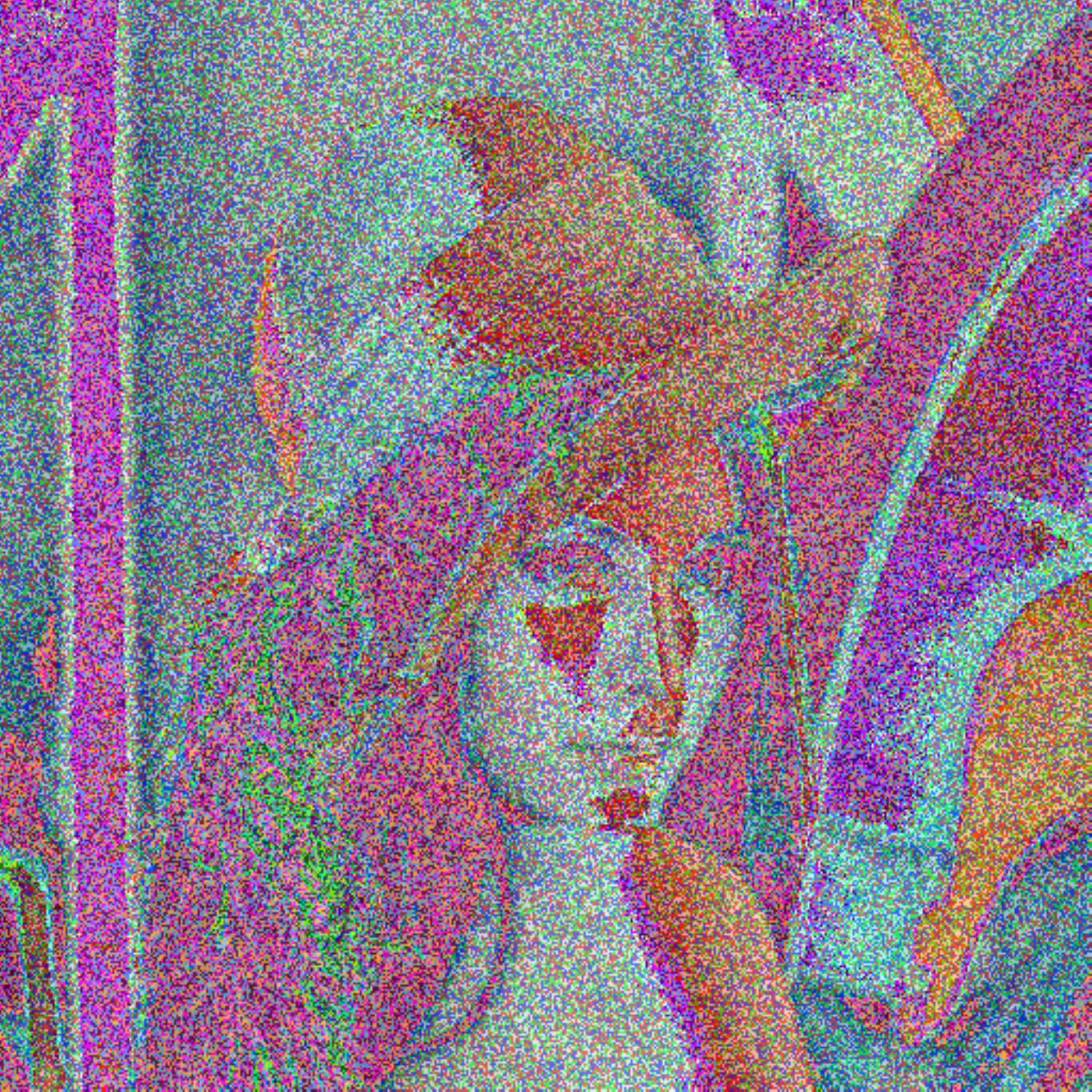}
c)
\end{minipage}
\caption{The decryption result with partially equivalent keys of
Class 2: a) the plain-image ``Lenna"; b) the cipher-image
corresponding to $K=``1A93DF25CF78\underline{DC44}E160"$; c) the
decryption result of subfigure b with a different key
$\widetilde{K}=``1A93DF25CF78\underline{5CC4}E160"$.}
\label{figure:EquivalentKey2}
\end{figure*}

Finally, it is worth mentioning that there exists an internal
relationship between the sub-images $\mathbb{I}_j$ and
$\mathbb{I}_{j+T/2}$, where $j\in\{0,\cdots,T/2-1\}$, which can be
easily deduced from the following fact about the updating process of
the subkeys: $K_i+K_{10}\cdot T/2=K_i+128\cdot
K_{10}/\gcd(K_{10},256)\equiv K_i+128=K_i\oplus 128\pmod{256}$.

\subsubsection{Reduction of the key space}

Based on the above analyses, we now summarize the influence of
invalid, weak and equivalent keys on the key space in
Table~\ref{table:keyspace}. According to the table, one can roughly
estimate that the size of key space is reduced to $2^{75}$, which is
somewhat smaller than $2^{80}$ (the one claimed in
\cite[Sec.~3.3]{Pareek:ImageEncrypt:IVC2006}).

\begin{table*}[htbp]
\center \caption{Reduction of the key space due to the existence of
invalid keys, weak keys and partially equivalent
keys.}\label{table:test}
\begin{tabular}{c|c|c}
\hline Subkeys & Size of reduced subkey-space & Reason\\
\hline\hline $K_1\sim K_3$ & - & $Y_0=0$\\
\hline $K_4\sim K_9$ & $2^{48}-5592406\approx 2^{48}$ & $X_0=0$\\
\hline $K_7\sim K_9$ & $136^3/2=2^{20.2624}$ & Equivalent key of Classes 1 and 2\\
\hline $K_{10}$ & $<(255-128-1)=126$ &
Weak keys about $K_{10}$\\
\hline
\end{tabular}\label{table:keyspace}
\end{table*}

\subsection{Guessing $K_{10}$ and $\{K_i\}_{i=1}^9$ separately}

The encryption process of the first block $I^{(16)}(0)$ depends only
on the secret values $Y_0$ and $K_{10}$. In other words, for the
first block one can consider $(Y_0,K_{10})$ as an equivalent to the
original key $K$. Then, by guessing the value of $(Y_0,K_{10})$ one
can get the value of $K_{10}$ with complexity $O(2^{32})$. Thus, the
other subkeys can be separately guessed with complexity $O(2^{72})$.
The total complexity of such an enhanced brute-force attack is
$O(2^{32}+2^{72})=O(2^{72})$, which is smaller than $O(2^{80})$, the
expected complexity of a simple brute-force attack.

\subsection{Guessing $K_{10}$ with a chosen plain-image}
\label{ssect:GuessingK10}

As remarked in Sec.~\ref{ssec:basicproperties}, all 16-pixel blocks
in $\mathbb{I}_j=\bigcup_{k=0}^{N_T-1}I^{(16)}(T\cdot k+j)$ are
encrypted with the same subkeys. If these blocks also correspond to
the same values of $Y_0$, then all the three encryption functions
for the R, G, B channels will become identical. Precisely, given two
identical blocks, $I^{(16)}(k_0)$ and $I^{(16)}(k_1)$, one can see
that the corresponding cipher-blocks will also become identical, if
the following two requirements are satisfied:
\begin{enumerate}
\item[(A)] the distance of the two blocks is a multiple of $T$,
i.e., $(k_0-k_1)\mid T$;

\item[(B)] $Y_0^{(k_0)}=Y_0^{(k_1)}$, where $Y_0^{(k_0)}$ and
$Y_0^{(k_1)}$ denote the values of $Y_0$ corresponding to the two
16-pixel blocks.
\end{enumerate}

Therefore, if the probability of the two cipher-blocks to be
identical is sufficiently large, one may use the distance between
them to determine the value of $T$ and narrow down the search space
of $K_{10}$.

It should be noted that the following two cases can both ensure the
requirement (B): 1) the sequences $\{P_j\}$ corresponding to the two
blocks are identical; 2) the sequences $\{P_j\}$ corresponding to
the two blocks are different (which may have $t\in\{0,\cdots,23\}$
identical elements), but the values of $Y_0$ are still identical.
The second case is tightly related to the ratio of 0-bits and 1-bits
in $B_2$. As an extreme example, when $B_2=0$ or $2^{24}-1$ (all the
bits of $B_2$ are 0 or 1), $B_2[P_j]$ will be fixed to be 0 or 1,
respectively. Assuming that the number of 1-bits in $B_2$ is $m$,
one can easily calculate the probability of
$B_2\left[P_j^{(k_0)}\right]=B_2\left[P_j^{(k_1)}\right]$ to be
$(m/24)^2+(1-m/24)^2$, and then the probability of
$Y_0^{(k_0)}=Y_0^{(k_1)}$ be $P_B=((m/24)^2+(1-m/24)^2)^{24}$. We
have carried out a large number of experiments to verify this
theoretical estimation and the results are shown in
Fig.~\ref{figure:probability_SameY0s_B}. In these experiments, all
possible values of $B_2$ were exhaustively generated to estimate the
probability (as the mean value) for $\min(m,24-m)\leq 4$, and
$\binom{24}{4}=10,626$ random keys were generated for
$\min(m,24-m)>4$.

\begin{figure*}[htbp]
\centering
\includegraphics[width=1.6\figwidth]{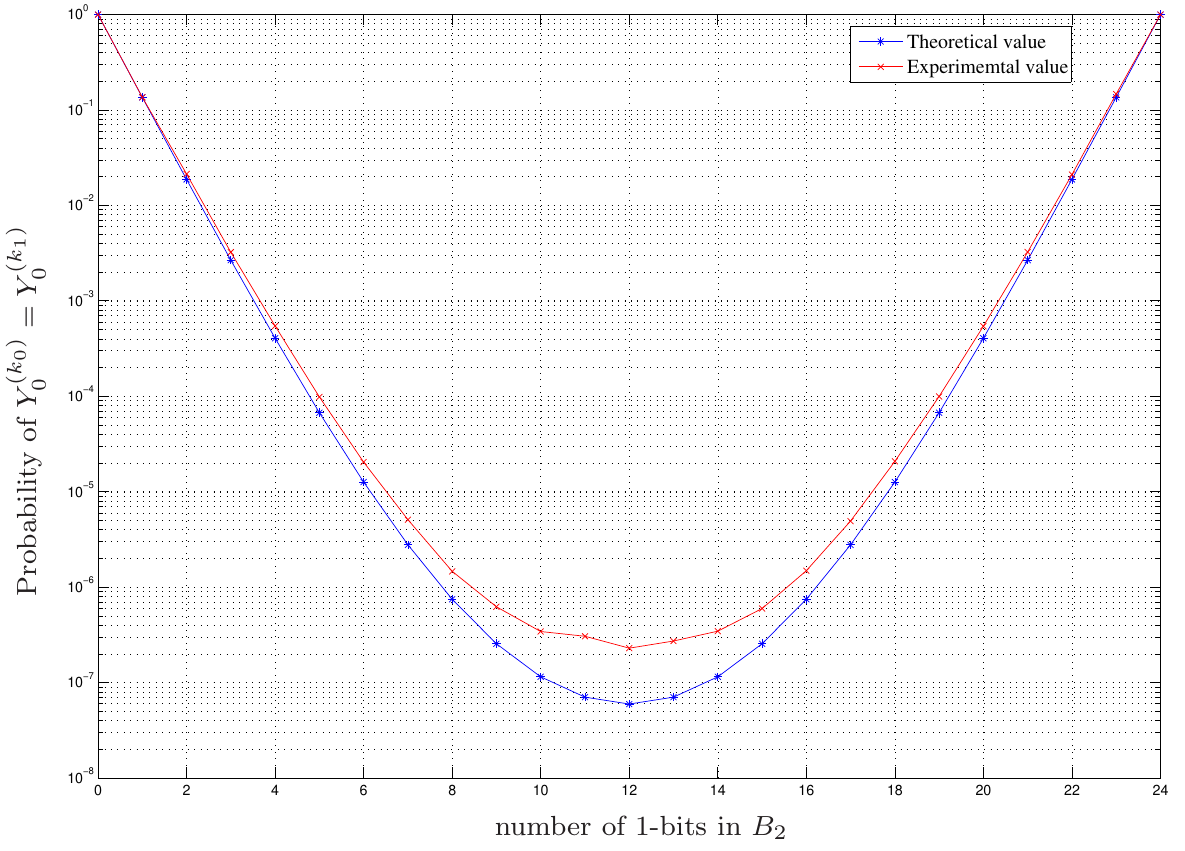}
\caption{Probability of $Y_0^{(k_0)}=Y_0^{(k_1)}$ with respect to
the number of 1-bits in $B_2$.}\label{figure:probability_SameY0s_B}
\end{figure*}

Since $Prob((k_0-k_1)\mid T)$ is $1/T$, the final probability that
both requirements hold is $P_B/T$. According to
Fig.~\ref{figure:probability_SameY0s_B}, this probability may be
large enough for an attacker to find some identical blocks in the
same set $\mathbb{I}_j$, especially when $\min(m,24-m)$ and $T$ are
both relatively small.

To show how the attack works, we chose a $512\times 512$ plain-image
in which all blocks are identical but all pixels in each block are
different from each other, and performed the attack for a secret key
$K=``2A84BCF35D70664E4740"$. As a result, we found 9 pairs of
identical blocks whose indices are listed in
Table~\ref{table:blockindex}. Because all these indices should
satisfy the requirement $(k_0-k_1)\mid T$, we can get an upper bound
of $T$ by solving their greatest common divisor of the differences
of the 9 indices. Thus, one immediately gets
$
\gcd(3161-1941,7083-2015,15255-3023,9163-4159,12113-5061,
16355-5507,12454-9166,12259-9655,13102-11090)=4.
$
This means $T\in\{2,4\}$, thus immediately leading to
$\gcd(K_{10},256)\in\{128,64\}$ and $K_{10}\in\{64,128,192\}$
according to Fact~\ref{fact:T_K10}. As can be seen, in this example
the size of the subkey space corresponding to $K_{10}$ is reduced
from 256 to 3, which is quite significant.

\begin{table*}[htbp]
\center\caption{The indices of 9 pairs of identical blocks in the
cipher-image corresponding to the plain-image of fixed value zero.}
\begin{tabular}{*{9}{c|}c}
\hline $k_0$ & 1941 & 2015 & 3023 & 4159 & 5061 & 5507 & 9166 & 9655
&
11090\\
\hline $k_1$ & 3161 & 7083 & 15255 & 9163 & 12113 & 16355 & 12454 &
12259 & 13102\\\hline
\end{tabular}
\label{table:blockindex}
\end{table*}

\subsection{Breaking $\{K_i\bmod 128\}_{i=4}^{10}$ with chosen-plaintext attack}
\label{ssec:CPA}

This subsection presents one of the most important results of this
work, since it shows how to partially break the encryption algorithm
using a very cost-effective chosen-plaintext attack, in which 128 or
even less plain-images are created. First, in
Sec.~\ref{sssec:Preliminaries} some mathematical devices are
introduced. Next, in Sec.~\ref{sssec:attack} the steps used to
recover subkeys $\{K_i\bmod 128\}_{i=4}^{10}$ are described in
detail. Finally some experimental results are given in
Sec.~\ref{subsubsection:experiments} validate the proposed attacks.

\subsubsection{Preliminaries}
\label{sssec:Preliminaries}

First, we prove some useful properties related to the composite
functions $E_i(x)$. These properties are essential for the attack to
be introduced below in this subsection.

\begin{theorem}\label{theorem:equivalentXor}
Let $F(x)=G_{2m+1}\circ\cdots\circ G_1(x)$ be a composite function
defined over $\{0,\ldots,2^n-1\}$, where $m,n\in\mathbb{Z}^+$,
$G_{2i}(x)=x\oplus\alpha_i$ for $i=1\sim m$,
$G_{2i+1}(x)=(x+\beta_i)\bmod 2^n$ for $i=0\sim m$ and
$\alpha_i,\beta_i\in\{0,\ldots,2^n-1\}$. If $F(x)=x\oplus \gamma$
for some $\gamma\in\{0,\ldots,2^n-1\}$, then
$\gamma\equiv\bigoplus_{i=1}^m\alpha_i\pmod{2^{n-1}}$.
\end{theorem}

\begin{proof}
Let $x=\sum_{j=0}^{n-1}x_j\cdot 2^j$,
$\alpha_i=\sum_{j=0}^{n-1}\alpha_{i,j}\cdot 2^j$,
$\beta_i=\sum_{j=0}^{n-1}\beta_{i,j}\cdot 2^j$, and
$F(x)=\sum_{j=0}^{n-1}F_j(x)\cdot 2^j$.

The proof is based on the following fact.

\textit{If $F$ verifies $F(x)=x\oplus \gamma$ for some
$\gamma=\sum_{j=0}^{n-1}\gamma_j\cdot 2^j$, then, for any $j=0\sim
n-1$, the result of the computation of $F_{j}(x)$ depends only on
the value of the $j$-th bit of $x$, that is, $x_{j}$. In other
words, the value of $F_j(x)$ is independent of $F_{j^*}$ if
$j^*\neq j$.}

We are going to check the computation of $F(x)$ starting from the
least significant bit. To get the value of $F_0(x)$, we only need to
calculate
$\widetilde{F_0}(x)=(\cdots((x_0+\beta_{0,0})\oplus\alpha_{1,0}+\beta_{1,0})
\oplus\cdots\oplus\alpha_{m,0}+\beta_{m,0})$, and then get the least
significant bit of $\widetilde{F_0}(x)$.\footnote{Here, $+\bmod 2^n$
is replaced by $+$ in the calculation process, because $\bmod 2^n$
does not affect any bit of $F(x)$.} Note that the carry bit
generated in each $+$ operation influences only the most significant
bits $F_1(x)\sim F_{n-1}(x)$, and for the least significant bit of
$\widetilde{F_0}(x)$ the operation $+$ is equivalent to $\oplus$.
Therefore, we immediately get
$F_0(x)=x_0\oplus\beta_{0,0}\oplus\alpha_{1,0}\oplus\beta_{1,0}\cdots\oplus\alpha_{m,0}\oplus\beta_{m,0}=
x_0\oplus(\alpha_{1,0}\oplus\cdots\oplus\alpha_{m,0})\oplus(\beta_{0,0}\oplus\cdots\oplus\beta_{m,0})$.

Then, let us study how the carry bits generated by $+$ operations
in the calculation on $\widetilde{F_0}(x)$ affect the value of
$F_1(x)$, as an effort to determine the value of
$\beta_{0,0}\oplus\cdots\oplus\beta_{m,0}$. Note the following two
facts about carry bits:
\begin{itemize}
\item
when $\beta_{i,0}=0$, no carry bit appears for any value of $x_0$;

\item
when $\beta_{i,0}=1$, a carry bit appears when $x_0=0$ or 1 after
the operation $+\beta_{i,0}$, and only for one possible value of
$x_0$ there will be a carry bit\footnote{To be more precise, if
there is a carry bit when $x_0=0$, then there will not be a carry
bit when $x_0=1$ and vice versa.}.
\end{itemize}
Denoting the number of $\beta_{i,0}$ whose value equals to 1 by
$N_0$, the above facts mean that $N_0$ can be obtained by counting
carry bits when $x_0=0$ and when $x_0=1$. That is,
$N_0=\sum_{x_0\in\{0,1\}}N_0(x_0)=N_0(0)+N_0(1)$, where $N_0(x_0)$
denotes the number of carry bits generated in the calculation
process on $\widetilde{F_0}(x)$ with respect to $x_0$.

The independence of $F_1(x)$ of $x_0$ means that $N_0(0)=N_0(1)$,
and as a result $N_0=N_0(0)+N_0(1)=2N_0(0)$ is an even number.
This immediately leads to the conclusion
$\beta_{0,0}\oplus\cdots\oplus\beta_{m,0}=0$. Thus, $F_0(x)=x_0
\oplus(\alpha_{1,0}\oplus\cdots\oplus \alpha_{m,0})$.

Next, consider $F_1(x)$. In this case,
$\widetilde{F_1}(x)=(\cdots((x_1+\beta_{0,1}+
CB_0(x_0))\oplus\alpha_{1,1}+\beta_{1,1}+
CB_1(x_0))\oplus\cdots\oplus\alpha_{m,1}+\beta_{m,1}+CB_m(x_0))$,
where $CB_i(x_0)$ denotes the bit carrying from $\widetilde{F_0}(x)$
during the $i$-th $+$ operation (which is equal to 0 when a carry
bit does not exist). Then, due to the same reason as in the case of
$F_0(x)$, we have $F_1(x)=x_1\oplus(\alpha_{1,1}\oplus\cdots\oplus
\alpha_{m,1})\oplus(\beta_{0,1}\oplus
CB_0(x_0)\cdots\oplus\beta_{m,1}\oplus CB_m(x_0))$. Observing the
expression of $\widetilde{F_1}(x)$, we can easily note the following
facts:
\begin{itemize}
\item
when $\beta_{i,1}=CB_i(x_0)=0$: no carry bit appears for any value
of $x_1$;

\item
when $\beta_{i,1}=CB_i(x_0)=1$: one carry bit always appears for any
value of $x_1$;

\item
when $\beta_{i,1}=0$, $CB_i(x_0)=1$, or when $\beta_{i,1}=1$,
$CB_i(x_0)=0$: one carry bit appears for only one value of $x_1$.
\end{itemize}
As a summary, only one carry bit may be generated from a pair of
$\beta_{i,1}$ and $CB_i(x_0)$, which means that one can consider
$\beta_{i,1}+CB_i(x_0)$ as a single value $\beta_{i,1}^*(x_0)$.

Denoting the number of $\beta_{i,1}^*$ whose value equals to 1 by
$N_1(x_0)$, the above facts imply that
$N_1(x_0)=\sum_{x_1\in\{0,1\}}N_1(x_0,x_1)=N_1(x_0,0)+N_1(x_0,1)$,
where $N_1(x_0,x_1)$ means the number of carry bits generated in
the calculation process on $\widetilde{F_1}(x)$ with respect to
$x_0$ and $x_1$. Then, because the value of $F_2(x)$ is
independent of $x_1$, we can get $N_1(x_0,0)=N_1(x_0,1)$ and
$N_1(x_0)$ is even. This means that $\beta_{0,1}\oplus
CB_0(x_0)\cdots\oplus\beta_{m,1}\oplus CB_m(x_0)=0$ and so
$F_1(x)=x_1\oplus(\alpha_{1,1}\oplus\cdots\oplus \alpha_{m,1})$.

The above deduction can be simply applied to other bits $F_2(x)\sim
F_{n-1}(x)$. As a result, we get
$F_i(x)=x_i\oplus(\alpha_{1,i}\oplus\cdots\oplus \alpha_{m,i})$,
$\forall i=0\sim n-1$.

Finally, combining all the cases together, we have the result that
$F(x)\equiv
x\oplus(\alpha_1\oplus\cdots\oplus\alpha_m)\pmod{2^{n-1}}$. This
means that $\gamma\equiv\bigoplus_{i=1}^m\alpha_i\pmod{2^{n-1}}$ and
the theorem is thus proved.
\end{proof}

\begin{corollary}\label{corollary:equivalentXor}
For the image encryption scheme under study, if there exists
$\gamma\in\{0,\ldots,255\}$ such that $E_i(x)=x\oplus \gamma$, then
$\gamma\in\left\{\bigoplus_i\alpha_i,\left(\bigoplus_i\alpha_i\right)\oplus
128\right\}$.
\end{corollary}
\begin{proof}
Consider the four classes of $E_i(x)$ as shown in
Sec.~\ref{ssec:basicproperties}.

\begin{enumerate}
\item
$E_i(x)=((\cdots((x\dotplus\beta_1)\oplus\alpha_1)\cdots)\oplus\alpha_{\lceil
(len-1)/2\rceil})\dotplus\beta_{\lceil len/2\rceil}$: From
Theorem~\ref{theorem:equivalentXor}, one has
$\gamma\in\left\{\bigoplus_{i=1}^{\lceil(len-1)/2\rceil}\alpha_i,\left(\bigoplus_{i=1}^{\lceil(len-1)/2\rceil}\alpha_i\right)\oplus
128\right\}$.

\item
$E_i(x)=((\cdots((x\dotplus\beta_1)\oplus\alpha_1)\cdots)\dotplus\beta_{\lceil
(len-1)/2\rceil})\oplus\alpha_{\lceil len/2\rceil}$: From
Theorem~\ref{theorem:equivalentXor}, one has $\alpha_{\lceil
len/2\rceil}\oplus\gamma\in\left\{\bigoplus_{i=1}^{\lceil
(len-1)/2\rceil}\alpha_i,\left(\bigoplus_{i=1}^{\lceil
(len-1)/2\rceil}\alpha_i\right)\oplus 128\right\}$, which means
$\gamma\in\left\{\bigoplus_{i=1}^{\lceil
len/2\rceil}\alpha_i,\left(\bigoplus_{i=1}^{\lceil
len/2\rceil}\alpha_i\right)\oplus 128\right\}$.

\item
$E_i(x)=((\cdots((x\oplus\alpha_1)\dotplus\beta_1)\cdots)\oplus\alpha_{\lceil
(len-1)/2\rceil})\dotplus\beta_{\lceil len/2\rceil}$: Assuming that
$x'=x\oplus\alpha_1$, we have
$E_i(x)=x\oplus\gamma=x'\oplus(\alpha_1\oplus\gamma)$. Then,
applying Theorem~\ref{theorem:equivalentXor} on $x'$, we can easily
get
$\alpha_1\oplus\gamma\in\left\{\bigoplus_{i=2}^{\lceil(len-1)/2\rceil}\alpha_i,\left(\bigoplus_{i=2}^{\lceil(len-1)/2\rceil}\alpha_i\right)\oplus
128\right\}$, thus
$\gamma\in\left\{\bigoplus_{i=1}^{\lceil(len-1)/2\rceil}\alpha_i,\left(\bigoplus_{i=1}^{\lceil(len-1)/2\rceil}\alpha_i\right)\oplus
128\right\}$.

\item
$E_i(x)=((\cdots((x\oplus\alpha_1)\dotplus\beta_1)\cdots)\dotplus\beta_{\lceil
(len-1)/2\rceil})\oplus\alpha_{\lceil len/2\rceil}$: Using a similar
process to the above class, one gets
$\gamma\in\left\{\bigoplus_{i=1}^{\lceil
len/2\rceil}\alpha_i,\left(\bigoplus_{i=1}^{\lceil
len/2\rceil}\alpha_i\right)\oplus 128\right\}$.
\end{enumerate}
The above four conditions together complete the proof of the
corollary.
\end{proof}

From Corollary~\ref{corollary:equivalentXor} and
Eq.~(\ref{eq:alpha_i_set}), we get the following result:
\begin{eqnarray}
\gamma\bmod 128 & = & \bigoplus\nolimits_i\alpha_i\bmod 128 \nonumber\\
               & \in &  \mathbb{A}^*=\{x\bmod 128\;|\;x
                \in  \mathbb{A}\cup\{0\}\}.\label{eq:gamma_mod128_set}
\end{eqnarray}

Assuming that $a_0^*=a_0\bmod 128$ and $a_1^*=a_1\bmod 128$, we have
\begin{equation}\label{eq:A_star}
\mathbb{A}^*=\{0,127,a_0^*,a_1^*,a_0^*\oplus 127,a_1^*\oplus
127,a_0^*\oplus a_1^*,a_0^*\oplus a_1^*\oplus 127\}.
\end{equation}
Observing the above equation, we can easily notice the following
facts:
\begin{enumerate}
\item
when $a_0^*=a_1^*\in\{0,127\}$, $\#(\mathbb{A}^*)=2$;

\item
when $a_0^*\in\{0,127\}$ and $a_1^*\not\in\{0,127\}$ (or
$a_1^*\in\{0,127\}$ and $a_0^*\not\in\{0,127\}$),
$\#(\mathbb{A}^*)=4$;

\item
when $a_0^*,a_1^*\not\in\{0,127\}$ and $a_0^*\oplus
a_1^*\in\{0,127\}$, $\#(\mathbb{A}^*)=4$;

\item
when $a_0^*,a_1^*\not\in\{0,127\}$ and $a_0^*\oplus
a_1^*\not\in\{0,127\}$, $\#(\mathbb{A}^*)=8$.
\end{enumerate}
Apparently, if we can get the set $\mathbb{A}^*$, it will be
possible to get the values of $a_0^*$ and $a_1^*$. The complexity of
such a process is summarized as follows:
\begin{enumerate}
\item
when $\#(\mathbb{A}^*)=2$, there are only 2 possible values of
$(a_0^*,a_1^*)$: (0,127) or (127,0);

\item
when $\#(\mathbb{A}^*)=4$, assuming that
$\mathbb{A}^*=\{0,127,a,a\oplus 127\}$, there are 8 possible values
of $(a_0^*,a_1^*)$: $(0,a)$, $(0,a\oplus 127)$, $(127,a)$,
$(127,a\oplus 127)$, $(a,a)$, $(a,a\oplus 127)$, $(a\oplus 127,a)$,
$(a\oplus 127,a\oplus 127)$;

\item
when $\#(\mathbb{A}^*)=8$, there are 24 possible values of
$(a_0^*,a_1^*)$: $a_0^*\in\mathbb{A}^*/\{0,127\}$ and
$a_1^*\in\mathbb{A}^*/\{0,127,a_0^*,a_0^*\oplus 127\}$.
\end{enumerate}
One can see that in any case the complexity is much smaller than
$2^7\times 2^7=2^{14}$, the complexity of exhaustively searching all
the bits of $a_0^*$ and $a_1^*$. This idea is the key for the
chosen-plaintext attack proposed in this subsection.

Next, let us find out how to distinguish XOR-equivalent encryption
functions. According to
Proposition~\ref{proposition:functionequivalent}, one can achieve
such a goal by checking the following 255 equalities: $F(x_1)\oplus
F(x_1\oplus i)=i$, where $x_1$ is an arbitrary integer in
$\{0,\ldots,255\}$ and $i=1\sim 255$.

\begin{proposition}
Let $F(x)$ be a function defined over $\{0,\ldots,2^n-1\}$, where
$n\in\mathbb{Z}^+$. Then, $F(x)=x\oplus\gamma$ for any
$x\in\{0,\ldots,2^n-1\}$ if and only if the following requirement
holds: there exists $x_1\in\{0,\ldots,2^n-1\}$ such that
$F(x_1)\oplus F(x_1\oplus i)=i, \forall
i\in\{1,\ldots,2^n-1\}$.\label{proposition:functionequivalent}
\end{proposition}
\begin{proof}
The ``only if" part is obvious. Now, let us prove the ``if" part.
Note that $F(x_1)\oplus F(x_1\oplus i)=i$ also holds when $i=0$. So,
when $i=x\oplus x_1$, we have $F(x_1\oplus x\oplus
x_1)=F(x)=F(x_1)\oplus x\oplus x_1=x\oplus (x_1\oplus F(x_1))$. When
$i=x_1$, we have $F(x_1)\oplus F(x_1\oplus x_1)=x_1$ and then get
$x_1\oplus F(x_1)=F(0)$. Therefore, $F(x)=x\oplus F(0)$, where
$F(0)=\gamma$ is a fixed value.
\end{proof}

For the encryption functions $E_i(x)$ composed of $\oplus$ and
$\dotplus$, the above result can be further simplified. From
Proposition~\ref{proposition:127eq255}, it is enough to check the
following 127 equalities: $F(x_1)\oplus F(x_1\oplus d)=d$, where
$x_1$ is an arbitrary integer in $\{0,\ldots,255\}$ and
$d\in\{1,\cdots,127\}$.

\begin{proposition}
Consider any encryption function $E_i(x)$ $(i=1\sim 3)$ defined in
Eqs.~(\ref{equation:encryptR})$\sim$(\ref{equation:encryptB}). If
there exists $x_1\in\{0,\ldots,255\}$ such that $E_i(x_1)\oplus
E_i(x_1\oplus d)=d$, $\forall d\in\{1,\ldots,127\}$, then
$E_i(x)=x\oplus E_i(0)$.\label{proposition:127eq255}
\end{proposition}
\begin{proof}
From Fact~\ref{corollary:xor128}, one has $E_i(x_1)\oplus
E_i(x_1\oplus 128)=128$ and $E_i(x_1)\oplus E_i(x_1\oplus j\oplus
128)=j\oplus 128$ for $j=1\sim 127$. This means that $E_i(x_1)\oplus
E_i(x_1\oplus j)=j$ holds $\forall j\in\{1,\ldots,255\}$. Then, from
Proposition~\ref{proposition:functionequivalent}, $E_i(x)=x\oplus
E_i(0)$.
\end{proof}

Next, let us investigate the probability that a given encryption
$E_i(x)$ is equivalent to $x\oplus\gamma$. Again, because the
theoretical analysis is quite difficult, we carried out a number of
random experiments with a $512\times 512$ plain-image for different
values of $K_{10}$, where $K_{1}\sim K_{9}$ were generated at
random. Generally speaking, this probability becomes smaller when
$K_{10}$ increases, but it fluctuates in a wide range for different
values of $K_1\cdots K_9$. Two typical examples are shown in
Fig.~\ref{figure:numberequivalentxor}, in which the XOR-equivalent
encryption functions involving the second kind of encryption
subfunctions (i.e., functions of the form $x\dotplus \beta$) and
those not involving these encryption subfunctions were counted
separately.

\begin{figure}[htbp]
\center
\begin{minipage}{1.2\figwidth}
\centering
\includegraphics[width=\textwidth]{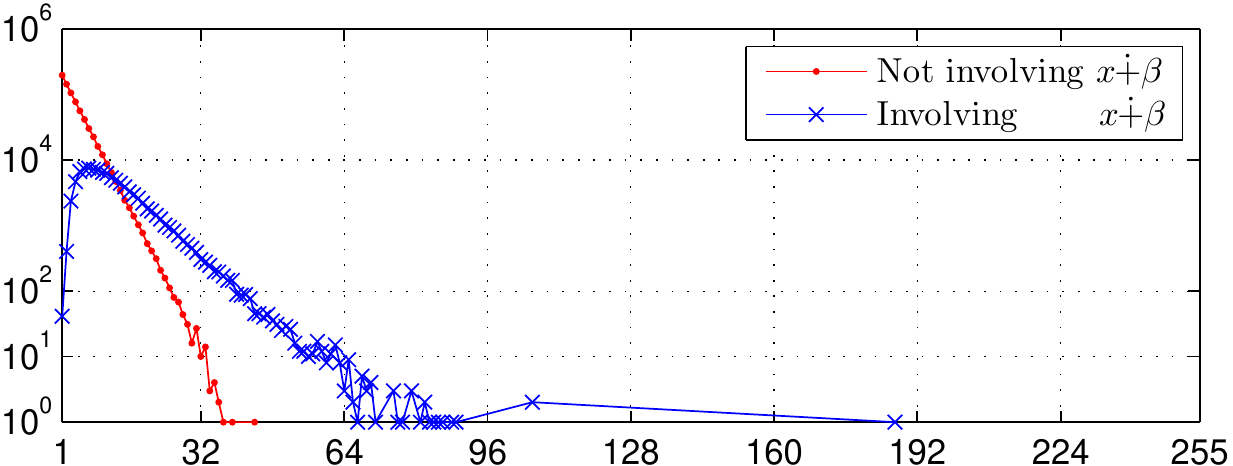}
a)
\end{minipage}\\
\begin{minipage}{1.2\figwidth}
\centering
\includegraphics[width=\textwidth]{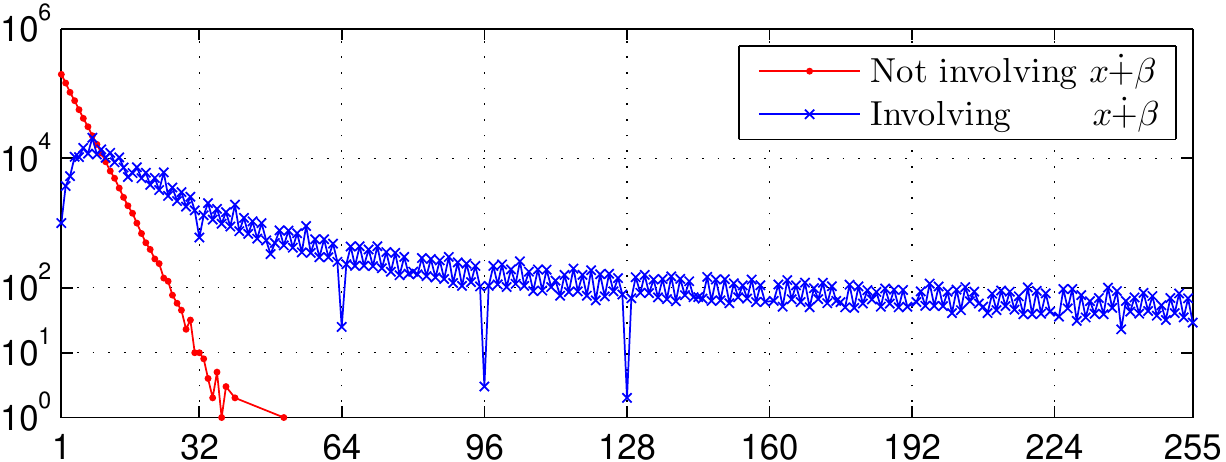}
b)
\end{minipage}
\caption{The number of pixels satisfying $E_1(x)=x\oplus \gamma$
under different values of $K_{10}$: a) $K_1\sim
K_9=``8DB87A1613D75ADF2D"$; b) $K_1\sim K_9=``2A84BCF35$
$D70664347"$.}\label{figure:numberequivalentxor}
\end{figure}

\subsubsection{Description of the attack}
\label{sssec:attack}

Based on the above discussions, a chosen-plaintext attack can be
developed by choosing 128 plain-images $\{I_l\}_{l=0}^{127}$ of size
$M\times N$ as follows: $I_l=I_0\oplus l$,\footnote{In this paper,
we use $I_l=I_0\oplus l$ to denote the following facts: $\forall
i=0\sim MN-1$, $R_l(i)=R_0(i)\oplus l$, $G_l(i)=G_0(i)\oplus l$ and
$B_l(i)=B_0(i)\oplus l$.} where $I_0$ can be freely chosen. To
facilitate the following description about the attack, denote the
encryption function $E_i(x)$ corresponding to the $j$-th pixel of
the $k$-th block by $E_{i,k,j}(x)$, and the parameters $a_0$, $a_1$
corresponding to the $k$-th block by $a_{0,i,k}$, $a_{1,i,k}$,
respectively. Similarly, for each updated subkey $K_j$, the value
corresponding to the $k$-th block is denoted by $K_{j,k}$. Then,
according to the discussion in
Sec.~\ref{sssec:PartiallyEquivalentKey2}, we have the following:
\begin{fact}
Given two XOR-equivalent encryption functions
$E_{i,k_1,j_1}(x)=x\oplus\gamma_{k_1,j_1}$ and
$E_{i,k_2,j_2}(x)=x\oplus\gamma_{k_2,j_2}$, if $k_1\equiv
k_2\pmod{T/2}$, then $\gamma_{k_1}\equiv\gamma_{k_2}\pmod{128}$.
\end{fact}

The proposed chosen-plaintext attack works in the following steps.

\textbf{Step 1} -- Finding XOR-equivalent encryption functions

For each color channel, scan the 128 plain-images to find encryption
functions $E_{i,k,j}$ that are equivalent to $x\oplus\gamma_k$,
where $\gamma_k=E_{i,k,j}(0)$ (according to
Proposition~\ref{proposition:127eq255}). Record all the
XOR-equivalent encryption functions corresponding to each color
channel in an $S_i\times 2$ matrix $\mymatrix{A}_i$, where $S_i$
denotes the number of blocks containing such encryption functions.
The first and the second rows of $\mymatrix{A}_i$ contain the block
indices and the corresponding values of $\gamma_k$, respectively.
Here, note that all XOR-equivalent encryptions in the same block are
identical, since they share the same parameters $a_{0,i,k}$ and
$a_{1,i,k}$.

The output of this step is composed of three matrices
$\{\mymatrix{A}_i\}_{1\leq i\leq 3}$, which require $\sum_{i=1}^3
2S_i$ memory units.

\textbf{Step 2} -- Estimating $\mathbb{A}_{i,k}^*$ (for each guessed
value of $K_{10}$)

Exhaustively search the value of $K_{10}$ and get the period
$T=256/\gcd(K_{10},256)$. Then, for each matrix $\mymatrix{A}_i$,
generate the following $T/2$ sets:
$\left\{\widetilde{\mathbb{A}}_{i,k}\right\}_{k=0}^{T/2-1}$, where
$\widetilde{\mathbb{A}}_{i,k}=\{\mymatrix{A}_i(s,2)\bmod 128|s\equiv
k\pmod{T/2}\}$. Next, expand each $\widetilde{\mathbb{A}}_{i,k}$ to
construct \[\widetilde{\mathbb{A}}_{i,k}^*=\left\{x_1\oplus x_2\oplus
x_3\left|\;x_1,x_2,x_3\in\widetilde{\mathbb{A}}_{i,k}\cup\{0,127\}\right.\right\},\]
which is an approximation of the following set:
$
\mathbb{A}_{i,k}^*=\{0,127,a_{0,i,k}^*,a_{1,i,k}^*,a_{0,i,k}^*\oplus
127,a_{1,i,k}^*\oplus 127, a_{0,i,k}^*\oplus
a_{1,i,k}^*,a_{0,i,k}^*\oplus a_{1,i,k}^*\oplus 127\},
$
where $a_{0,i,k}^*=(a_{0,i,0}+k\cdot K_{10})\bmod 128$ and
$a_{1,i,k}^*=(a_{1,i,0}+k\cdot K_{10})\bmod 128$. Note that
$a_{0,i,0}$ and $a_{1,i,0}$ are the two subkeys corresponding to the
color channel in question.

Then, if there exists $k\in\{0,\cdots,T/2-1\}$ such that
$\#\left(\widetilde{\mathbb{A}}_{i,k}^*\right)\not\in\{2,4,8\}$, one
can immediately conclude that the current value of $K_{10}$ is wrong
and then remove it from the list of candidate values for $K_{10}$.

The output of this step includes a list of $N$ candidate values of
$K_{10}$ and at most $3T/2$ sets
$\{\widetilde{\mathbb{A}}_{i,k}\}_{1\leq i\leq 3 \atop 0\leq k\leq
T/2-1}$ for each candidate value of $K_{10}$. The total number of
memory units required is not greater than $6\times 3NT/2=9NT\leq
12\times 256\times 128=294912\approx 2^{18.2}$, which is practical
for a PC to store the intermediate data. Here, note that 0 and 127
are always in $\mathbb{A}^*$, so they do not need to be kept.

\textbf{Step 3} -- Determining $\{K_i\bmod 128\}_{i=4}^{10}$

For each color channel, choosing the set
$\widetilde{\mathbb{A}}_{i,k_0}^*$ of the greatest size\footnote{The
greatest size may be 8, 4 or 2. When it is 4 or 2,
$\widetilde{\mathbb{A}}_{i,k_0}^*$ may not be a good estimation of
$\mathbb{A}_{i,k_0}^*$ and as a result cannot be used to support the
attack. This case often occurs when $K_{10}$ is relatively large,
thus leading to a very small occurrence probability of
XOR-equivalent encryption functions (see
Fig.~\ref{figure:numberequivalentxor}).}, one can exhaustively
search all possible values of $(a_{0,i,k_0}^*,a_{1,i,k_0}^*)$, i.e.,
search all possible values of $a_{0,i,0}^*=(a_{0,i,k_0}^*-k_0\cdot
K_{10})\bmod 128$ and $a_{1,i,0}^*=(a_{1,i,k_0}^*-k_0\cdot
K_{10})\bmod 128$. Note that $a_{0,1,0}^*=K_4\bmod 128$ and
$a_{1,1,0}^*=K_7\bmod 128$ (red channel), $a_{0,2,0}^*=K_5\bmod 128$
and $a_{1,2,0}^*=K_8\bmod 128$ (green channel),
$a_{0,3,0}^*=K_6\bmod 128$ and $a_{1,3,0}^*=K_9\bmod 128$ (blue
channel).

All the guessed values of $(a_{0,i,0}^*,a_{1,i,0}^*)$ are verified
by employing the relationship between $\mathbb{A}_{i,k_0}^*$ and
other sets $\{\mathbb{A}_{i,k}^*\}_{k\neq k_0}$. If all possible
values of $(a_{0,i,0}^*,a_{1,i,0}^*)$ are eliminated, the current
value of $K_{10}$ can also be eliminated. Note that the other three
values of a valid candidate $(a_{0,i,0}^*,a_{1,i,0}^*\oplus
128,K+{10}\bmod 128)=(u,v,w)$ will also pass the verification
process due to Fact~\ref{fact:equivalnetkey} below: $(u\oplus 127,
v\oplus 127, 128-w)$, $(v, u, w)$, and ($v\oplus 127, u\oplus 127,
128-w$).

\begin{fact}
Given $x, a, c\in \{0,\cdots,127\}$, $x+ac\equiv(x\oplus
127+(128-a)c)\oplus 127\pmod{128}$. \label{fact:equivalnetkey}
\end{fact}

The output of this step is a list of candidate values of
\[
K^*=(K_4\bmod 128,\cdots,K_9\bmod 128,K_{10}\bmod 128).
\]
In the worst case, the number of all possible values is $N\times
24^3\leq 256\times 24^3=3538944\approx 2^{21.6}$, which is still
much smaller than the number of all possible values of the subkey
$K^*$: $2^{6\times 7+8}=2^{50}$. In the best case, the number of
candidate values is only $2\times 2^3=16$ (according to
Fact~\ref{fact:equivalnetkey}).

\subsubsection{Experimental Results}
\label{subsubsection:experiments}

To validate the feasibility of the above attack, we have carried out
a real attack with a randomly-generated secret key
$K=``2A84BCF25E6A664E4C41"$. As a result, we got the following
output from Step 2:
\begin{eqnarray*}
K_{10} & \in & \{1,3,\cdots,255\},\\
\mathbb{A}_{0,6}^* & = & \{0, 127,108, 20, 7, 107, 120, 108\},\\
\mathbb{A}_{0,28}^* & = & \{0, 127, 115, 125, 14, 0, 12, 113\},\\
\mathbb{A}_{0,79}^* & = & \{0, 127, 116, 117, 1, 10, 11, 126\},\\
\mathbb{A}_{1,19}^* & = & \{0, 127, 16, 33, 49, 111, 94, 78\},\\
\mathbb{A}_{1,28}^* & = & \{0, 127, 106, 122, 21, 5, 111, 16\},\\
\mathbb{A}_{2,7}^* & = & \{0, 127, 19, 78, 108, 49, 34, 93\},\\
\mathbb{A}_{2,18}^* & = & \{0, 127, 34, 93, 3, 33, 124, 94\}.
\end{eqnarray*}
The final output of the attack (i.e., the output of Step 3) is shown
in Table \ref{table:recoveredkey}.

\renewcommand\tabcolsep{5pt}

\newcolumntype{C}[1]{>{\centering\let\newline\\\arraybackslash\hspace{0pt}}m{#1}}

\begin{table}[htbp]
\center\caption{The final output of a real attack, where the
underlined data form the real values of $\{K_i\bmod
128\}_{i=4}^{10}$.}
\begin{tabular}{C{1.7cm}|*{5}{c|}c}
\hline
\multirow{2}{1in}{$K_{10}\bmod 128$}     & \multicolumn{6}{c}{$\{K_i\bmod 128\}_{i=4}^{9}$}\\
\cline{2-7} & $i=4$ & $i=7$ & $i=5$ & $i=8$ & $i=6$ & $i=9$\\
\hline
\multirow{8}{*}{63}& \multirow{4}{*}{25} & \multirow{4}{*}{13} & \multirow{2}{*}{33} & \multirow{2}{*}{49 } & 51 & 21\\
\cline{6-7} & & & & & 21 & 51\\
\cline{4-7} & & &\multirow{2}{*}{49} & \multirow{2}{*}{33} & 51 & 21\\
\cline{6-7} & & & & & 21 & 51\\
\cline{2-7} & \multirow{4}{*}{13} & \multirow{4}{*}{25} & \multirow{2}{*}{33} & \multirow{2}{*}{49 } & 51 & 21\\
\cline{6-7} & & & & & 21 & 51\\
\cline{4-7} & & & \multirow{2}{*}{49} & \multirow{2}{*}{33} & 51 & 21\\
\cline{6-7} & & & & & 21 & 51\\
\hline \multirow{8}{*}{\underline{65}} & \multirow{4}{*}{102} & \multirow{4}{*}{114} & \multirow{2}{*}{94} & \multirow{2}{*}{78} & 76 & 106\\
\cline{6-7} & & & & & 106& 76\\
\cline{4-7} & & & \multirow{2}{*}{78} & \multirow{2}{*}{94} & 76 & 106\\
\cline{6-7} & & & & & 106& 76\\
\cline{2-7} & \multirow{4}{*}{\underline{114}} & \multirow{4}{*}{\underline{102}} & \multirow{2}{*}{94} & \multirow{2}{*}{78} & 76 & 106\\
\cline{6-7} & & & & & 106& 76\\
\cline{4-7} & & & \multirow{2}{*}{\underline{78}} & \multirow{2}{*}{\underline{94}} & 76 & 106\\
\cline{6-7} & & & & & \underline{106} & \underline{76}\\
\hline
\end{tabular}
\label{table:recoveredkey}
\end{table}

Finally, note that one may also be able to distinguish some
XOR-equivalent encryption functions even with less than 128 chosen
plain-images. To investigate such a possibility, we have carried out
some experiments by choosing the following $(n+1)<128$ plain-images
instead: $\{I_l\}_{l=0}^n$, where $I_l=I_0\oplus l$ for any $l>0$.
Let $N(n)$ be the number of XOR-equivalent encryption functions
detected with the above $n+1$ chosen plain-images. The ratio
$r(n)=N(127)/N(n)$ gives an estimation of the probability that a
detected XOR-equivalent encryption function is real. For three
randomly-generated keys, the values of $r(n)$ with respect to
different values of $n$ are shown in
Fig.~\ref{figure:n_chosen_plaintexts}, from which one can see that
the value of $r(n)$ always increases significantly when $n$
increases from $2^i-1$ to $2^i$ ($i=1\sim 6$). We also carried out
experiments for other random keys, and found out that this fact
holds for most of them. According to this experimental result, we
can choose the following 13 plain-images to minimize the number of
chosen plaintexts: $I_0$, $I_1=I_0\oplus 1$, $I_2=I_0\oplus 2$,
$I_3=I_0\oplus 3$, $I_4=I_0\oplus 4$, $I_5=I_0\oplus 7$,
$I_6=I_0\oplus 8$, $I_7=I_0\oplus 15$, $I_8=I_0\oplus 16$,
$I_9=I_0\oplus 31$, $I_{10}=I_0\oplus 32$, $I_{11}=I_0\oplus 63$ and
$I_{12}=I_0\oplus 64$. Then, for 1,000 randomly-generated secret
keys, our experiments show that the average value of
$r^*=N(127)/N^*$ is about 0.825, where $N^*$ denotes the number of
detected XOR-equivalent encryption functions with the 13 chosen
plain-images. Note that the value of $r^*$ is not accurate when
$N^*$ is too small. If only those keys corresponding to $N^*\geq
100$ are considered, the average value of $r^*$ increases to about
0.9234. If only those corresponding to $N(n)\geq 1000$ are counted,
the average value of $r^*$ becomes about 0.9826. In practice, one
may have to use more than 13 chosen plain-images to run the proposed
attack, but it is expected that $O(20)$ chosen plain-images are
enough in most cases.

\begin{figure}[htbp]
\centering
\includegraphics[width=1.5\figwidth]{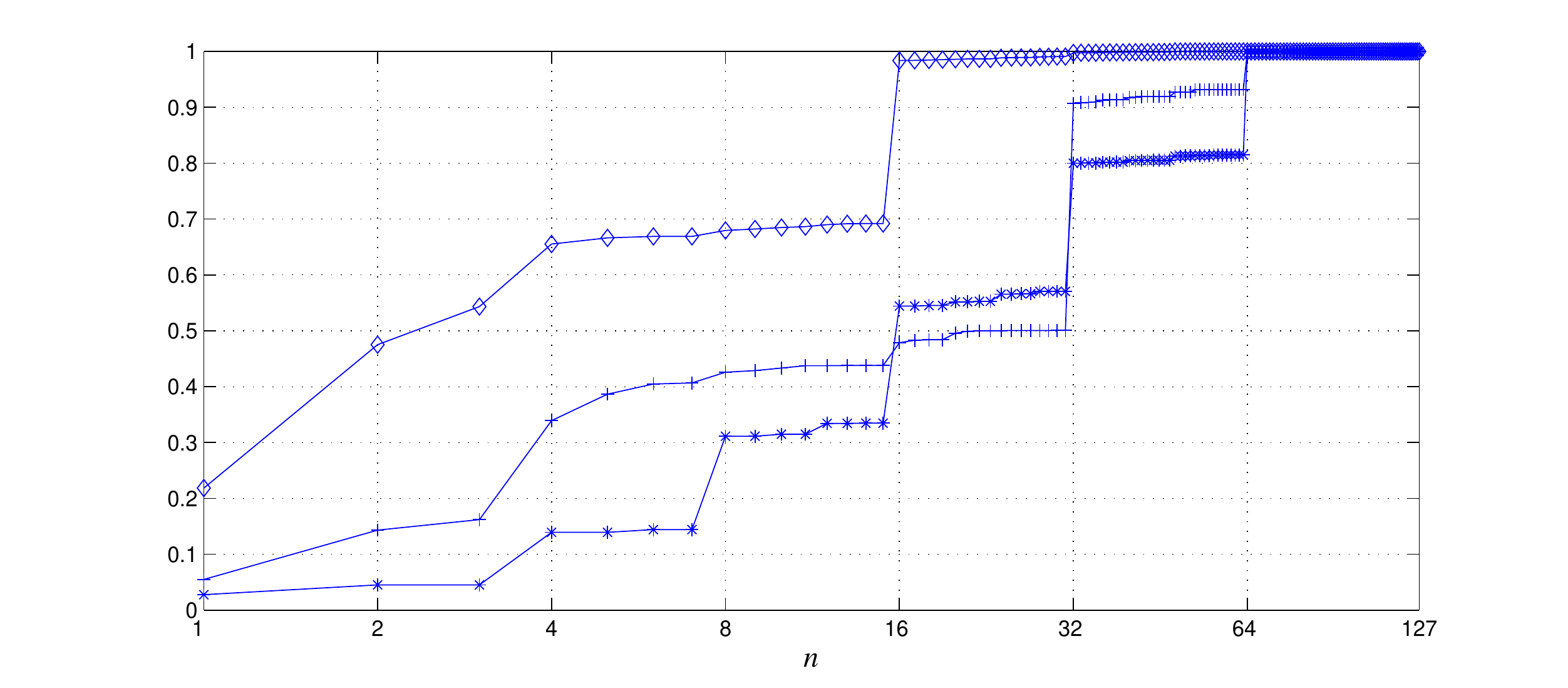}
\caption{The values of $r(n)$ with respect to different values of
$n=1\sim 127$, where the three lines correspond to the results of
three randomly-generated keys.}\label{figure:n_chosen_plaintexts}
\end{figure}

\subsection{Known-plaintext attack based on a masking image}

According to the results shown in
Fig.~\ref{figure:numberequivalentxor}, we know that many encryption
functions are equivalent to XOR operations. Therefore, if we
consider all the encryption functions as XOR-equivalent ones, then a
masking image can be obtained by simply XORing a known plain-image
and the corresponding cipher-image pixel by pixel. By using this
masking image as an equivalent of the secret key to decrypt other
cipher-images, all the pixels encrypted by real XOR-equivalent
encryption functions will be correctly recovered. If the number of
such correctly-recovered pixels is sufficiently large, some visual
information about the plain-images may be obtained. It is expected
that this known-plaintext attack can work well when $K_{10}$ is
relatively small. Figure~\ref{figure:knownplaintextattack} shows two
examples of this attack when $K_{10}=6$ and 30, from which one can
see that some important visual information about the plain-image is
revealed.

\begin{figure}[htbp]
\centering
\begin{minipage}[t]{\imgwidth}
\centering
\includegraphics[width=\textwidth]{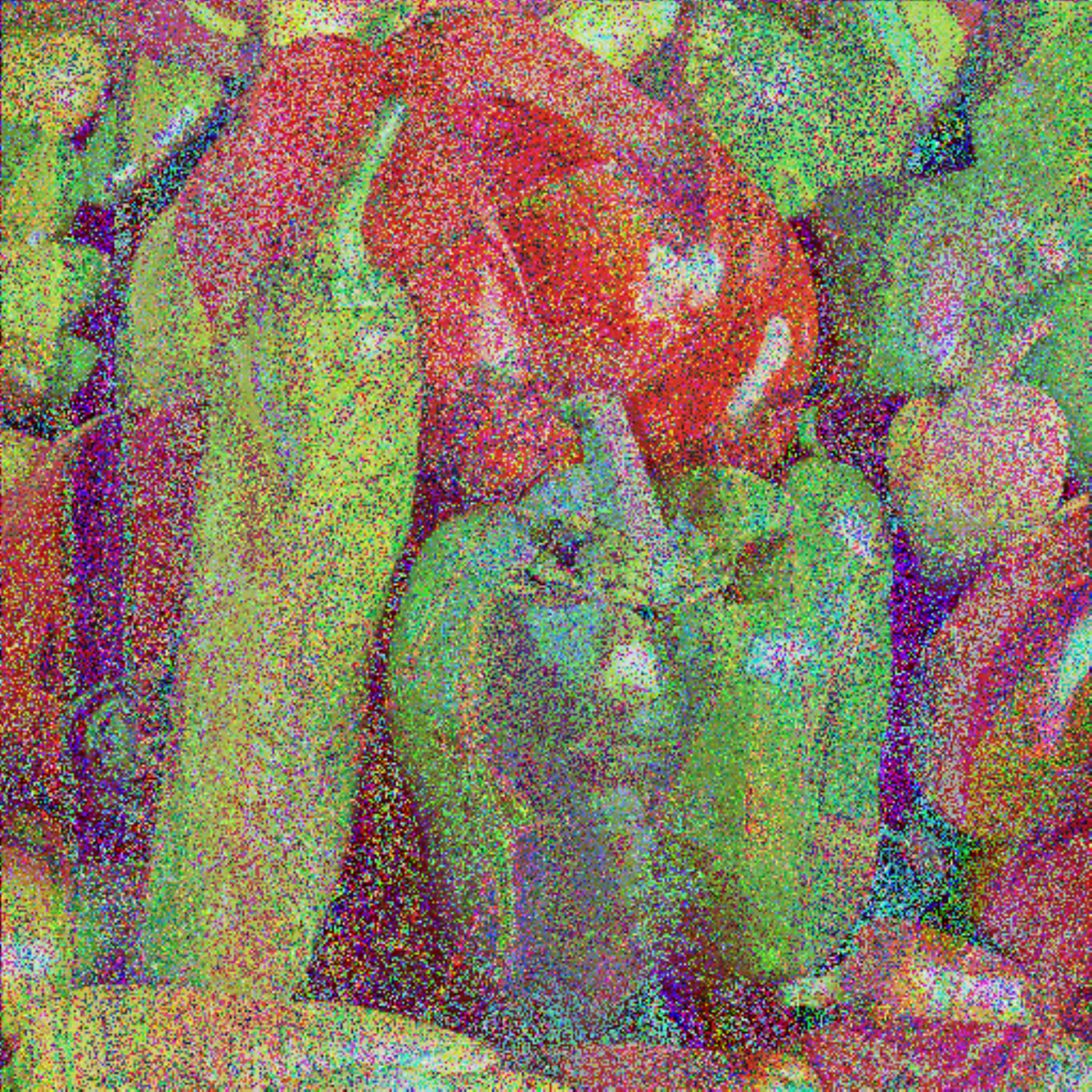}
a)
\end{minipage}
\begin{minipage}[t]{\imgwidth}
\centering
\includegraphics[width=\textwidth]{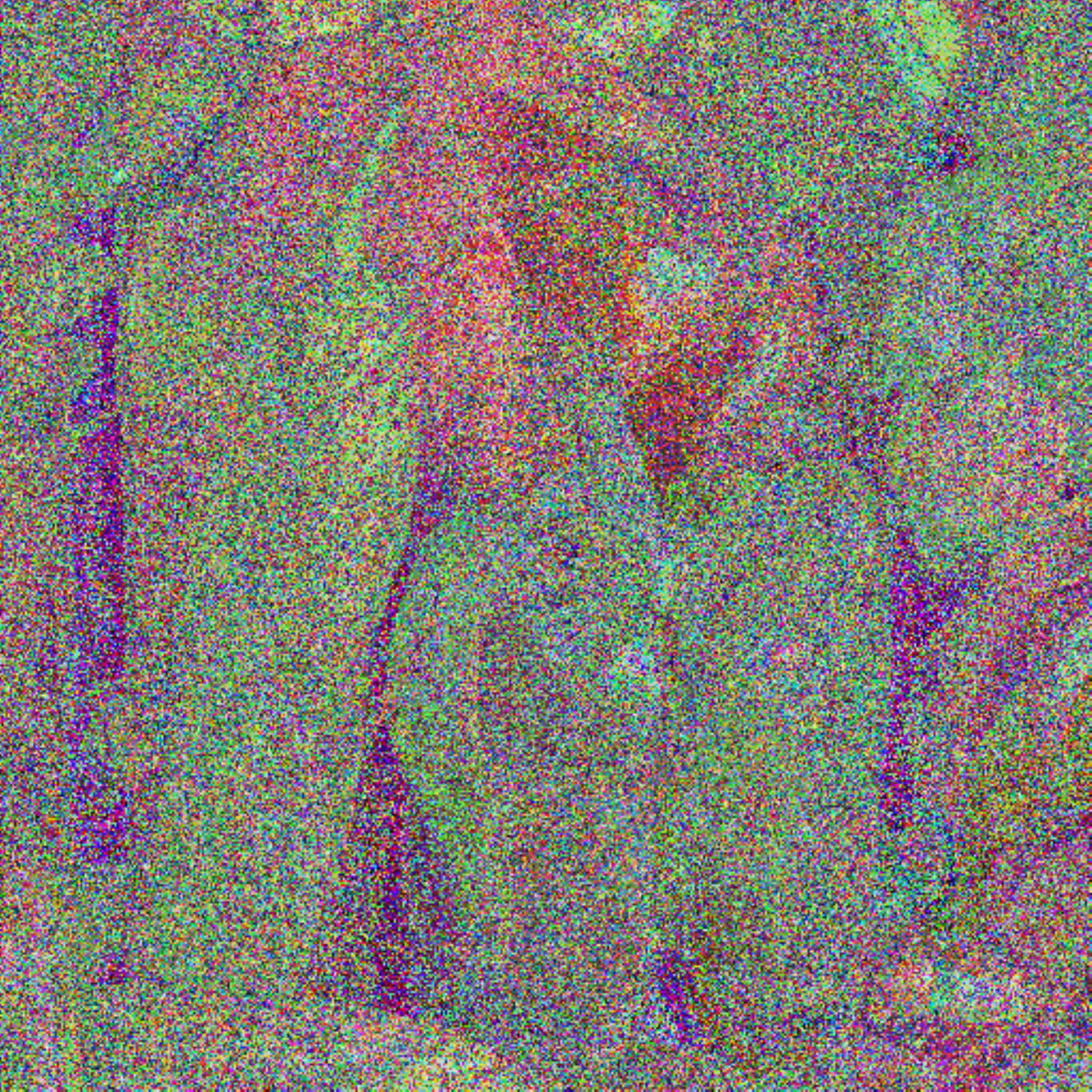}
b)
\end{minipage}
\caption{The result of breaking a plain-image ``Peppers" with a
masking image obtained when ``Lenna"
(Fig.~\ref{figure:EquivalentKey2}a) is the known plain-image: a)
$K=``8DB87A1613D75ADF2D06"$; b) $K=``8DB87A1613D75ADF2D1E"$.}
\label{figure:knownplaintextattack}
\end{figure}

\section{Conclusion}

In this paper, the security of a recently-proposed image encryption
scheme has been analyzed in detail. It is found that there exists a
number of invalid keys, weak keys and partially equivalent keys,
which reduce the size of the key space. Some attacks to a number of
subkeys have also been developed: 1) given a chosen plain-image, a
subkey $K_{10}$ can be guessed with a complexity less than $2^8$; 2)
part of the key may be recovered with a chosen-plaintext attack
using at most 128 chosen plain-images. The scheme under study can
also be broken with only one known plain-image, when the subkey
$K_{10}$ is small. In addition, some other insecure problems about
the scheme have been discussed throughout. The cryptanalysis
presented in this paper shed some new light on attacking other
encryption schemes that are composed of multi-round encryption
functions, a relatively difficult but important topic to be further
investigated in the near future.

\section*{Acknowledgement}

Chengqing Li was partially supported by the Research Grants Council
of the Hong Kong SAR Government under Project 523206 (PolyU
5232/06E). Shujun Li was supported by the Alexander von Humboldt
Foundation, Germany. Juana Nunez and Gonzalo Alvarez were partially
supported by Minis\-terio de Educaci\'on y Ciencia of Spain,
Research Grant SEG2004-02418.

\bibliographystyle{elsart-num}
\bibliography{IVC,lcq_paper_bib}

\end{document}